\documentclass[aps,floatfix,amsmath,twocolumn,notitlepage,nofootinbib,superscriptaddress,pra]{revtex4-2}
\usepackage{enumitem}

\usepackage[toc,page,header]{appendix}

\usepackage[ruled]{algorithm2e}
\SetKwFor{For}{for (}{) $\lbrace$}{$\rbrace$}
\usepackage{amssymb}
\usepackage{graphicx}
\usepackage{graphics}
\usepackage{amsmath}
\usepackage{amsthm}
\usepackage{color}
\usepackage{dsfont}
\usepackage{wrapfig}

\usepackage{longtable}

\newcommand{\ketbra}[2]{|#1\rangle\!\langle#2|}

\usepackage{mathtools}
\usepackage{centernot}
\usepackage[colorlinks=true, linkcolor=blue, citecolor=red]{hyperref}
\usepackage{verbatim}
\usepackage{xcolor}

\usepackage{multirow}
\usepackage{booktabs} 
\usepackage{tabularx}
\usepackage{makecell}
\usepackage{wrapfig}
\usepackage{array}
\usepackage{pifont}

\usepackage{thmtools}
\usepackage{thm-restate}

\newcolumntype{Y}{>{\centering\arraybackslash}X}
\newcolumntype{C}[1]{>{\centering\let\newline\\\arraybackslash\hspace{0pt}}m{#1}}

\newcommand{\impliesandnot}{%
  \mathrel{%
    \vcenter{\offinterlineskip
      \ialign{##\cr$\impliedby$\cr\noalign{\kern-1.5pt}$\centernot\implies$\cr}%
    }%
  }%
}

\usepackage{color}

\def\>{\rangle}
\def\<{\langle}

\renewcommand{\geq}{\geqslant}
\renewcommand{\leq}{\leqslant}

\newtheorem{theorem}{Theorem}
\newtheorem*{theorem*}{Theorem}
\newtheorem{corollary}{Corollary}
\newtheorem{lemma}{Lemma}
\newtheorem*{lemma*}{Lemma}

\newtheorem{definition}{Definition}[section]
\newtheorem*{definition*}{Definition}

\newtheorem{example}{Example}
\theoremstyle{definition}

\theoremstyle{remark}

\theoremstyle{definition}

\newtheorem*{axiom*}{Axiom}

\newcommand{\xtens}{\mathbin{\mathop{\otimes}\limits_{\max}}}
\newcommand{\ntens}{\mathbin{\mathop{\otimes}\limits_{\min}}}

\newcommand{\tr}{{\rm Tr}}

\newcommand{\ve}[1]{{\left\vert\kern-0.25ex\left\vert\kern-0.25ex\left\vert #1 
    \right\vert\kern-0.25ex\right\vert\kern-0.25ex\right\vert}}

\newcommand{\ket}[1]{|#1\rangle}

\newcommand{\Det}{\mathsf{Det}}

\definecolor{cool_green}{rgb}{0.0, 0.5, 0.0}

\begin{document}

\title{No extension of the Quantum Tensor Product admits a Superposition principle}

\author{Vincenzo Fiorentino}
\email{vincenzo.fiorentino@york.ac.uk}
\affiliation{Department of Mathematics, University of York, York YO10 5GH, United Kingdom}

\author{Kuntal Sengupta}
\thanks{This work was partially conducted at the author's previous affiliations: the University of York and Institute Néel, CNRS. \href{kuntal.sengupta@neel.cnrs.fr}{kuntal.sengupta@inria.fr}}
\affiliation{Univ. Grenoble Alpes, Inria, 38000 Grenoble, France}

\begin{abstract}
Textbook quantum superposition refers to the feature that certain linear combinations of Hilbert space rays, each representing a valid quantum state, are themselves valid states. This notion is not operational, and it relies on the underlying Hilbert space formalism. Recent proposals for experimental tests of indefinite causal order, as well as tests probing the non-classicality of gravity, pivot on superposition, thereby calling for a theory-independent, operational formalisation of the concept. Here, we define superposition within the framework of Generalised Probabilistic Theories, based on observed statistics in prepare-and-measure experiments. Using this, we formulate three superposition principles to investigate which structural features of quantum theory carry over to other theories. We study conditions under which these principles carry over from subsystems to their compositions; to this end, we show that the quantum tensor product emerges as the largest composition rule for quantum systems respecting all three principles. Furthermore, we show how non-classical features such as entanglement and preparational uncertainty can be viewed as special forms of superposition.

\end{abstract}

\maketitle

\section{Introduction}

Textbook quantum mechanics represents quantum states as rays of a Hilbert space. That every such ray can be expressed as a complex linear combination of rays representing other pure states is referred to as \textit{quantum superposition}. It plays a central role in contrasting quantum theory from classical theory and has been at the heart of celebrated quantum algorithms, kicking off the field of quantum computing. In recent years, experiments based on superposition have been proposed to test non-classicality gravity~\cite{marletto_gravitationally_2017,PhysRevLett.119.240401}. Current experimental certification of the lack of definite causal order within quantum theory also relies on quantum superposition~\cite{procopio2015,rubino2017,Rubino2020,Rubino2022experimental,guo2020experimental,Cao2022semiDI,Zhu2023,Qu_2025,goswami2018communicating,yin2023experimental,Richter_2025}. 
However, what (if any) objective description a quantum state given by a ray in a Hilbert space corresponds to remains unclear, and consequently, so does that of a complex linear combination of such rays. 
In fact, in the absence of the Hilbert space formalism, which observational aspects one should attribute to the presence of superposition in quantum theory remains a crucial open question in quantum foundations, answering which will not only allow one to explore superposition in theories beyond quantum mechanics, but also investigate whether superposition can be used as a necessary principle to reconstruct quantum theory.

A quantum state can be viewed as a superposition of a collection of other states whenever there exists a measurement that can perfectly discriminate this collection of states but returns probabilistic outcomes with respect to the first. This \textit{relational} property can be phrased in the framework of Generalised Probablistc Theories (GPTs), where every experiment can be broken down into state preparations, transformations and measurements. 
The existence of states in a GPT fulfilling this relational property, hence displaying the same statistical behavior in the state-discriminating measurement described above as a quantum state in a superposition of other quantum states, defines what we call the existence of superposition in a GPT. We classify a variety of GPTs based on whether they admit this operational notion of superposition. 

This differs from previously introduced notions of superposition in that it features superposition as a specific type of non-classicality beyond non-simpliciality, is relational, unlike in~\cite{aubrun_entanglement_2022}, and it is applicable to any GPT, unlike in~\cite{dariano_classical_2020}, which requires infinitely many extremal states.

Given that superposition exists, one can ask whether it can be used to identify structural properties of a theory. For instance, in quantum theory, (i) all extremal states can be expressed as a superposition of some other states, or (ii) given any measurement basis, any extremal state that is not deterministic with respect to it is in superposition of the collection of extremal states that are perfectly discriminated by it, or (iii) if an extremal state is a superposition of some collection of extremal states, every state in that collection is also a superposition of another collection that includes the original state.  Inspired by these structural properties of quantum theory, we introduce three \textit{principles} of superposition: \textit{complete}, \textit{uniform}, and \textit{mutual}, and show that they do not form a hierarchy. One might ask if these principles can play a role in the reconstruction of quantum theory on more physical grounds.
We manage to provide an affirmative answer to this at the level of the composition rule: any non-signalling composition of quantum subsystems strictly larger than the quantum tensor product will not admit the principle of mutual superposition. This principle, therefore, forbids any strict non-signalling extension of the quantum tensor product. 

Our operational formulation also reveals a common structure underlying superposition, preparational uncertainty and entanglement, the latter two of which are typically studied separately. We show that preparational uncertainty necessarily implies superposition and, under suitable sharpness assumptions, leads directly to stronger superposition principles. For the class of composite theories considered here, entanglement likewise entails superposition, with extremal entangled states appearing as superpositions relative to perfectly distinguishable product states. These results recast two central signatures of non-classicality as particular manifestations of operational superposition.

This paper is organised as follows. Section.~\ref{Sec::Preliminaries} provides a brief review of the GPT framework and introduces basic concepts. Section.~\ref{Sec::SuperpositionQT} discusses superposition in quantum theory which is then generalised to GPTs in Section.~\ref{Sec::SuperpositionGPTs}. Section.~\ref{Sec::SuperpositionCompositeSystems} focuses on how a composite system inherits superposition from its subsystems and and how entanglement implies superposition. In Section~\ref{Sec::MaximalityofQuantumTensorProduct} we follow on to single out the quantum tensor product from larger composition rules. In Section.~\ref{Sec::RelationToUncertainty} we show how preparational uncertainty can be viewed as a special form of superposition. Finally, Section.~\ref{Sec:Discussion} provides a summary and an outlook. All proofs are provided in the Appendix.

\section{Preliminaries} \label{Sec::Preliminaries}
GPTs are a framework to analyse theories in which experiments can be decomposed into preparations, transformations and measurements of systems. In a GPT, the set of states, i.e.\ the \textit{state space}, $\mathcal{S}$, is given by a compact convex subset of a real vector space $\mathbb{V}$, such that for some fixed element $u \in \mathbb{V}^*$, $ u(s) = 1$, for all $s \in \mathcal{S}$, where $\mathbb{V}^*$ is the dual vector space. The \textit{effect space}, $\mathcal{E}$, is a compact convex subset of $\mathbb{V}^*$ with the property that for every $e \in \mathcal{E}$, $e(s)\coloneqq \left<e,s\right> \in [0,1]$ for every $s \in \mathcal{S}$,  where $\left<\cdot,\cdot\right>$ is an inner product on $\mathbb{V}$ and $u \in \mathcal{E}$. In addition, for every $e \in \mathcal{E}$, $u - e \in \mathcal{E}$. When $\mathcal{E}$ is the largest subset of $\mathbb{V}^*$ compatible with these conditions, the GPT is said to satisfy the \textit{no-restriction hypothesis}. We will denote the largest effect space compatible with the state space $\mathcal{S}$ as $\mathcal{E_S}$. A state (effect) is said to be \textit{extremal} if it cannot be expressed as a convex combination of other states (effects). We denote by $\mathrm{Ext}[\mathcal{S}]$  and $\mathrm{Ext}[\mathcal{E}]$ the sets of extremal states and effects, respectively.

Given two state spaces $\mathcal{S}_1 \subset \mathbb{V}_1$ and $\mathcal{S}_2 \subset \mathbb{V}_2$, a bipartite state space, $\mathcal{S}_1 \boxtimes \mathcal{S}_2 \subset \mathbb{V}_1 \otimes \mathbb{V}_2$ (shorthanded to $\mathcal{S}_{1,2} $), is a collection of states identified with a composition rule $\boxtimes$ with the property that for every $s \in \mathcal{S}_{1,2}$, $u_1 \otimes \mathrm{id}_2 (s) \in {\mathcal{S}_{2}}_{\leq}$ and $\mathrm{id}_1 \otimes u_2   (s) \in {\mathcal{S}_{1}}_{\leq}$, where $u_1,\mathrm{id}_1$ and $u_2,\mathrm{id}_2$ are the unit effect and identity map pair on the first and second systems respectively and $\leq$ denotes the subnormalised state space characterised by the convex hull of a state space with the zero vector. Two composition rules of interest are the \textit{minimal} and \textit{maximal} tensor products. The minimal composition is defined as 
$$
 \mathcal{S}_1\ntens \mathcal{S}_2 \coloneqq \mathrm{ConvHull}\{s_1 \otimes s_2|s_1 \in \mathcal{S}_1, s_2 \in \mathcal{S}_2
    \}.
$$
The maximal $\mathcal{S}_1  \xtens \mathcal{S}_2$ composition is the largest bipartite state space compatible with all effects in the minimal composition of the respective effect spaces and satisfying normalisation. A bipartite compostion rule $\boxtimes$ is said to be non-signalling if for any two state spaces $\mathcal{S}_1$ and $\mathcal{S}_2$, $\mathcal{S}_1\ntens \mathcal{S}_2 \subseteq \mathcal{S}_{1,2} \subseteq \mathcal{S}_1\xtens \mathcal{S}_2$.

A collection $\mathsf M \coloneqq \{e_j\}_j$ of effects is said to be a \textit{measurement} iff $\sum_j e_j = u$. A collection of extremal states $\mathsf{D} \coloneqq \{s_i\}_i$ is said to be \textit{perfectly discriminable} (PD) if and ony if there exists a measurement $\{e_j\}_j \subseteq \mathrm{Ext}[\mathcal{E}]$, such that $\left<e_j,s_i\right> = \delta_{i,j}$, where $\delta_{i,j} = 1$ when $i=j$ and 0 otherwise. $\mathsf{D}$ is said to be a \textit{maximal} PD (MPD) set iff for any other PD set of states $\mathsf{D}'$, $|\mathsf D| \geq |\mathsf D'|$ holds, where $|\mathsf D|$ denotes the cardinality of the set $\mathsf D$. The cardinality of an MPD set is known as the \textit{measurement} or \textit{operational dimension} of the system~\cite{brunner2014dimension, sen_2022_timelike}. For classical theory the operational dimension coincides to the number of extremal states. A measurement $\mathsf M \subseteq \mathrm{Ext}[\mathcal{E}]$ that perfectly discriminates an MPD set $\mathsf{D}$ and such that $|\mathsf{M}|=|\mathsf{D}|$ is said to be a \textit{maximally distinguishing and extremal} (MDE) measurement. We denote by $\mathsf{Det(M)}$ the set of all extremal states, $s$, such that $e(s)=1$ for some $e \in \mathsf{M}$.

In this work, we assume that the set of MDE measurements is \textit{informationally complete}; that is, every state in the theory can be tomographically reconstructed using only the statistics obtained from MDEs.

Qubit quantum theory can be phrased as a GPT by identifying its state space with the set of $2 \times 2$ unit-trace positive semi-definite matrices, a compact convex subset of the real vector space, $\mathbb{H}(\mathbb{C}^2)$, of $2 \times 2$ Hermitian matrices. An  effect, $\Pi$, is a $2 \times 2$ positive semi-definite matrix, such that $\mathds{1}-\Pi$ is also positive semi-definite, where $\mathds{1}$ is the $2 \times 2$ identity matrix. The effect space is the largest set of such matrices with $\mathds{1}$ as the unit effect, and forms a compact convex subset of  $\mathbb{H}(\mathbb{C}^2)$. The action of the effects on the states is given by the Hilbert-Schmidt inner product.  Qudit quantum theory can be understood in a similar way by identifying its state space with the compact convex subset, of $d \times d$ unit-trace positive semi-definite matrices, of the real vector space of $d \times d$ Hermitian matrices, and so on. The operational dimension of a qudit is equal to $d$ \cite{sen_2022_timelike}. Given two quantum systems whose state spaces are the set of unit-trace positive semi-definite matrices in the vector spaces $\mathbb{H}(\mathbb{C}^{d_1})$ and $\mathbb{H}(\mathbb{C}^{d_2})$, the composite state space is given by the set of unit-trace positive semi-definite matrices in the vector space $\mathbb{H}(\mathbb{C}^{d_1d_2})$. We will refer to this as the \textit{quantum tensor product} composition of two quantum systems.

In this work we investigate a variety of GPTs other than quantum and classical theory. We consider \textit{gbits} which are characterised by two fiducial measurements with two outcomes each ($p(0|0),p(0|1)$)---the probability of the second outcome can be trivially calculated and is redundant. The vertices of the state space of gbits correspond to $(0,0)$, $(0,1)$, $(1,0)$ and $(1,1)$. Systems characterised by more fiducial measurements with two or more number of outcomes are called \textit{g-systems}. The vertices of the state spaces for g-systems are obtained in a similar way as gbit state spaces. The minimal tensor product of g-systems is called Generalised Local Theory (GLT)~\cite{PhysRevA.75.032304}. The maximal tensor product of such systems is called Generalised Non-Signalling Theory~\cite{PhysRevA.75.032304} or more commonly Boxworld~\cite{Short_2010}. We also consider single-system state spaces whose geometry is given by regular polygons, providing approximations to the real qubit state space, i.e.\ the Bloch disc~\cite{Janotta_2011}. In addition, we also construct single-system state spaces corresponding to fragments of quantum theory. We call them GPT-1 and GPT-2 and outline their details in Appendix~\ref{Appendix::GPT1_2}.

\section{Superposition in Quantum Theory} \label{Sec::SuperpositionQT}

An extremal (pure) qubit state $\Psi$ is a rank-1 $2 \times 2$ matrix and therefore $\ket{\Psi}\rangle \coloneqq \ketbra{\psi}{\psi}$, where $\ket{\psi} \in \mathbb{C}^2$ is its eigenvector with unit eigenvalue. Similarly, every extremal qudit state can also be represented by a vector in $\mathbb{C}^d$. The notion of superposition in quantum theory is inherently connected to this mathematical property. Given a finite-dimensional Hilbert-space $\mathcal{H} \cong \mathbb{C}^d$, and collection of extremal linearly independent qudit states $\mathsf{D} = \{\ket{\Phi_i}\rangle \}_i^{d' \leq d}$, an extremal qudit with state $\ket{\Psi}\rangle $ is said to be in superposition of states in $\mathsf{D}$ whenever there exists a collection of complex numbers $\{c_i\}_{i=1}^{d'}$, with $c_i \notin \{0,1\}$ for any $i$ and $\sum_{i=1}^{d'} |c_i|^2=1$, such that
\begin{equation}  \label{Equation::TextbookSuperposition}
    \ket{\psi}= \sum_{i=1}^{d'} c_i \ket{\phi_i} \ . 
\end{equation}  
In fact, for any $\ket{\psi}$, the pair of collection $\mathsf{D}$ and $\{c_i\}_{i=1}^{d'}$ can always be found with $d'=d$. Textbook superposition is pertained to quantum states admitting Eq.~\eqref{Equation::TextbookSuperposition}, often associated to the analogy of superposition in classical wave mechanics. However, unlike waves in classical wave mechanics, vectors of the form $\ket{\psi}$ cannot be observed. Therefore, it is natural to ask whether the notion of superposition in quantum mechanics is an artifact of the Hilbert space formalism or one can associate observable properties to it. In this work, we indicate that it is possible to operationally define superposition as constraints on the statistics observed in an experiment, in a theory-independent way. If one then assumes that the underlying theory is quantum mechanics, Eq.~\eqref{Equation::TextbookSuperposition} can be derived.

To do this, let us analyse the textbook superposition in quantum theory phrased as a GPT. Consider an experiment where one needs to perfectly discriminate states in the set $\mathsf{D}=\{\ket{\Phi_i}\rangle \}_{i=1}^d$. This is indeed possible whenever $\tr[\ket{\Phi_i}\rangle \ket{\Phi_j}\rangle] = \delta_{i,j}$, using the measurement $\mathsf{M}=\mathsf{D}$. In addition, $\mathsf{D}$ is an MPD set and $\mathsf{M}$ its MDE measurement. Then, for any state $\ket{\Psi}\rangle$ such that $\mathsf{c}_i \coloneqq \tr[\ket{\Phi_i}\rangle \ket{\Psi}\rangle] \in (0,1)$ for all $i$, we can find complex numbers $\{c_i\}_i$ with the property that $|c_i|^2= \mathsf{c}_i$ and Eq.~\eqref{Equation::TextbookSuperposition} holds. Since $\{\mathsf c _i\}_i$ represents a probability distribution, $\sum_{i=1}^{d} |c_i|^2=\sum_{i=1}^{d} \mathsf{c}_i =  1$. In fact, every $\ket{\Psi}\rangle \notin \mathsf{D}$ admits this. Quantum superposition, therefore, is a property of the state $\ket{\Psi}\rangle$ with respect to the MPD set $\mathsf{D}$.

An operational view of the scenario above is the following: Suppose an experimenter has access to a measurement device $\mathsf{M}=\{ e_i \}_{i=1}^d$ that can perfectly distinguish some MPD set $\mathsf{D} = \{s_j\}_{j=1}^d$, i.e., $e_i(s_j)=\delta_{ij}$. Another device prepares a state $r$ which is observed to be probabilistic with respect to this measurement. If this preparation device were sampling from the set $\{s_1,\dots,s_d\}$, the observed statistics would have a classical explanation as a statistical mixture. However, if the experimenter can rule out this possibility, for instance, by certifying that $r$ is an extremal state, then there must be a different ``phenomenological explanation'' for this probabilistic behaviour. In quantum theory, one would call this  superposition.

\section{Superposition in GPTs} \label{Sec::SuperpositionGPTs}

The operational prescription presented above allows us to define superposition in a theory-independent way, thus extending its applicability beyond quantum mechanics to arbitrary GPTs. We formalise superposition as a statistical property of an extremal state relative to an MPD set--—that is, as a \textit{relational property} of a state with respect to other states. This aligns with how superposition is generally regarded in standard quantum theory, where the state $\ket{\psi}$ in Eq.~\eqref{Equation::TextbookSuperposition} is understood as a superposition \textit{of the states} $\{\ket{\phi_i}\}_i$. Previous attempts at abstracting away from the Hilbert space formalism are either not operational~\cite{aubrun_entanglement_2022} or do not apply to arbitrary GPTs~\cite{dariano_classical_2020}. In~\cite{sengupta2024achievingmaximalcausalindefiniteness}, the author proposed an operational notion of when one can infer that a GPT admits superposition. However, this notion is insufficient to capture the relational aspect mentioned above. By contrast, we provide a criterion on the statistics of a prepare-and-measure experiment which, in standard quantum theory, is satisfied if there exist a collection of extremal states and an extremal state outside this collection that is in a \textit{superposition} of the collection. Such a prepare-and-measure scenario can be defined in any theory, allowing us to abstract this relational property, and define superposition to arbitrary theories.

\begin{definition}[Superposition] 
    Let $(\mathcal{S},\mathcal{E})$ be a pair of state and effect spaces. Let $\mathsf{D} \coloneqq\{s_i\}_{i=1}^n$ be an MPD set and $\mathsf{M} \coloneq \{e_i\}_{i=1}^n$ an MDE measurement of $\mathsf{D}$. Let $r \in \mathrm{Ext}[\mathcal{S}]$ be a state with the property that for some $\mathsf{M}' \subseteq \mathsf{M}$, $e_j (r) \in (0,1)$ holds for all $e_j \in \mathsf{M}'$, and $\sum_{e_j \in \mathsf{M}'}  e_j( r) =1$. Then, $r$ is a superposition of the states in the set $\mathsf {D}' \coloneqq \{s \in \mathsf D | e_j(s) =1, e_j \in \mathsf{M}' \}$. 
  
    \label{Def::Superposition}
\end{definition}

Note first, that if we drop the requirement that $\mathsf{D}$ is a maximal set, then although one could infer that $r$ is a superposition, the set of states with respect to which it is a superposition may be falsely concluded. The following qutrit example captures this: the non-maximal PD set $\mathsf{D}= \{\ket{0}\rangle, \ket{2}\rangle\}$ is perfectly discriminated by the degenerate measurement $\mathsf{M}=\{\ket{0}\rangle, \mathbb{I}_3 - \ket{0} \rangle \}$. Without the maximality requirement, a state $\ket{\psi}\rangle$ such as $\ket{\psi} = \alpha \ket{0} + \beta \ket{1} + \gamma \ket{2}$ with $\alpha,\beta,\gamma \neq 0$ would oddly be labeled as a superposition of $\ket{0}\rangle$ and $\ket{2}\rangle$ \textit{only}, despite its non-zero inner-product with $\ket{1}\rangle$. Second, superposition is defined in this work only for extremal states. This is because superposition is only attributed to extremal states in quantum theory. We discuss in Appendix~\ref{Appendix::PlavalaSup} a notion of superposition in the absence of extremality and its consequences.

Rather than having a relational property between one state and a collection of states, a next natural question is whether this relation can be used to identify structural properties of a theory. To do this, we first notice that given the axioms of textbook quantum mechanics, the following statements become equivalent: (i) all pure states can be expressed as a superposition of some other pure states,  (ii) given any measurement basis, any pure state that is not deterministic with respect to it is in superposition of the collection of pure states that are perfectly discriminated by it, (iii) if a pure state is a superposition of some collection of pure states, every state in that collection is also a superposition of another collection that includes the original state. Despite of being equivalent for quantum theory, they might pertain to inequivalent structural features for arbitrary GPTs. Inspired from this observation we formulate three \textit{principles} of superposition that encode facts about superposition applicable to the whole theory.

\begin{definition}[Existence and Principles of Superposition]
    Let $(\mathcal{S},\mathcal{E})$ be a state and effect pair of a GPT. Superposition \textbf{exists} iff some $s \in \mathrm{Ext}[\mathcal{S}]$ is a superposition, following Definition~\ref{Def::Superposition}. Further, this superposition is
     \begin{itemize}
        \item \textbf{complete}  iff every $s \in \mathrm{Ext}[\mathcal{S}]$ is a superposition;
        \item \textbf{uniform}  iff for any MPD set $\mathsf{D}$, every $s \in \mathrm{Ext}(\mathcal{S}) \setminus \mathsf{D}$ is a superposition of some $\mathsf{D}'_s \subseteq \mathsf{D}$;
        \item \textbf{mutual} iff whenever $s \in \mathrm{Ext}[\mathcal{S}]$ is a superposition of $\mathsf{D}'$, every $r \in \mathsf{D}'$ is a superposition of some fixed $\mathsf{D}''$, such that $s \in \mathsf{D}''$.  
        
    \end{itemize}
    \label{Def::SuperpositionPrinciples}
\end{definition}

Definitions~\ref{Def::SuperpositionPrinciples} apply to both single and composite systems. 
Table \ref{Tab::SuperpositionPrinciples} summarises the superposition properties of several theories, with the mathematical derivations deferred to Appendix \ref{Appendix: superposition GPT examples}.

\begin{table}[ht]
    \centering
    
  \includegraphics[width=1.0\linewidth]{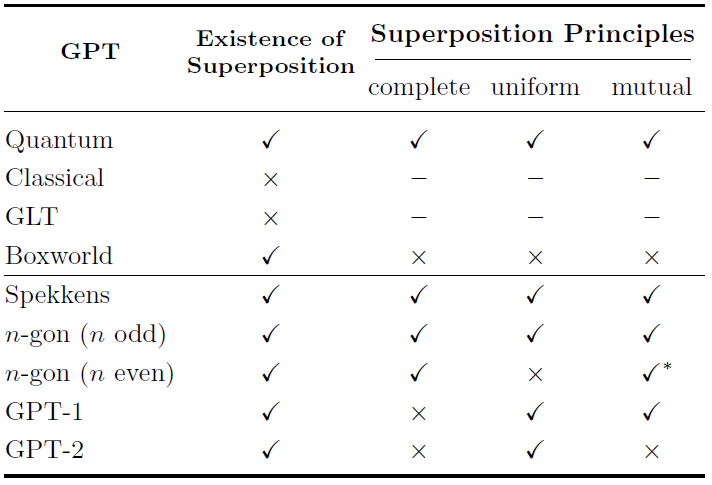}
    \caption{Classification of superposition across various GPTs. Theories above the horizontal line include both single and composite systems, whereas for those below the line we only examine individual systems. Here, $n$-gon refers to regular polygons with $n>4$ sides and $\checkmark^*$ denotes that the property does not hold when $n=6$. GPT-1 and GPT-2 are models defined in Appendix \ref{Appendix::GPT1_2}. The symbol $\textbf{--}$ indicates that the principles are not applicable.} 
    \label{Tab::SuperpositionPrinciples}
\end{table}

 Classical theory, characterised by a simplicial state space and no-restriction hypothesis, fails to exhibit any form of superposition. However, the lack of classicality is not a signature of existence of superposition. Single system gbit state spaces or their minimal tensor product composition (GLT) do not admit any form of superposition.
 Interestingly, Spekken's toy theory \cite{spekkens_evidence_2007}, introduced to motivate  the classicality of interference phenomena with epistemic restrictions, admits all superposition principles. In addition, superposition can exist in composite systems even when it is is absent in the subsystems, as captured by Boxworld where, however, none of the three stronger princples hold. A more general discussion on superposition in composite systems is presented in Section.~\ref{Sec::SuperpositionCompositeSystems}.

The table also demonstrates that the three types of superposition do not form a strict hierarchy. One can construct GPTs that admit uniform and mutual superposition but that do not admit complete superposition (see GPT-1 in Appendix~\ref{Appendix::GPT1_2}), or that only admit uniform superposition (see GPT-2 in Appendix~\ref{Appendix::GPT1_2}).

Furthermore, the $n$-gon models display a non-quantum feature of superposition as per Definition~\ref{Def::Superposition}, where a state can be a superposition of a set while remaining perfectly discriminable from some of its members (see Appendix~\ref{Appendix:RegularNgons}). We also observe that real qubit quantum theory inherits the full suite of superposition principles from the class of odd $n$-gons (inner approximations of the Bloch disc). In contrast, even $n$-gons (outer approximations of the Bloch disc) fail to uphold uniform superposition, establishing a parity-dependent structural gap that vanishes only in the continuous limit. Whether there is a higher principle underlying such explanation is not known to us.

\section{From Subsystems to Arbitrary Compositions} \label{Sec::SuperpositionCompositeSystems}

Superposition, as defined in Definition~\ref{Def::Superposition}, requires an MDE and an extremal state that is probabilistic with respect to its outcomes. Now, if two systems which (at least one of them) admit superposition are composed, a natural question is whether superposition will exist in this composition as well. One therefore first needs to characterise the MDEs of the composition. In quantum theory, the operational dimension of a composite system $(\mathcal{S_Q}_{1,2}, \mathcal{E_Q}_{1,2})$ equals the product of the operational dimensions of its components (see Lemma~\ref{lemma::multiplicityquantum} in Appendix~\ref{Appendix: quantum tensor product and superposition} for a more general statement). Consequently, the tensor product $\mathsf{M}_1 \otimes \mathsf{M}_2 \coloneqq \{ e_i \otimes f_j \}_{ij}$ of the local MDEs $\mathsf{M}_1 = \{ e_i \}_i$ and $\mathsf{M}_2 = \{ f_j \}_j$ constitutes a valid MDE for the composite system. A sufficient condition for this property was provided in~\cite{sen_2022_timelike}, with the minimal and maximal products of quantum systems serving as key examples. However, it does not hold in general. A counterexample is any non-signalling tensor product of the state spaces of two regular pentagons (see Appendix~\ref{Appendix::MinProductPentagons}).

Whenever product MDEs are MDEs for the composite system, i.e.\ given that the operational dimension of the composition is the product of the operational dimensions of its subsystems, we show that existence of superposition and the principle of complete superposition carries over from subsystems to composites. For uniform superposition to carry over, we needed to assume \textit{outcome sharpness of MDEs}: for every effect $e$ in an MDE there exists a unique extremal state $s$ such that $e(s)=1$, which holds in both quantum and classical theories.

\begin{restatable}{lemma}{SuperpositionFromSubsystems}
    Let $(\mathcal{S}_1, \mathcal{E}_{\mathcal{S}_1})$ and $(\mathcal{S}_2, \mathcal{E}_{\mathcal{S}_2})$ be two GPT systems with operational dimensions $d_1$ and $d_2$ respectively. Let $(\mathcal{S}_{1,2},\mathcal{E}_{\mathcal{S}_{1,2}})$ be a no-sigalling composition such that the operational dimension is $d_1d_2$. Then,
    \begin{enumerate}[label=(\roman*), font=\normalfont]
        \item If $(\mathcal{S}_1,\mathcal{E}_{\mathcal{S}_1})$ satisfies the existence of superposition, then $(\mathcal{S}_{1,2},\mathcal{E}_{\mathcal{S}_{1,2}})$ also satisfies the existence of superposition;
        \item If $(\mathcal{S}_1,\mathcal{E}_{\mathcal{S}_1})$ satisfies complete superposition, then $(\mathcal{S}_{1,2},\mathcal{E}_{\mathcal{S}_{1,2}})$ also satisfies complete superposition;
        \item If  $(\mathcal{S}_1,\mathcal{E}_{\mathcal{S}_1})$, $(\mathcal{S}_2,\mathcal{E}_{\mathcal{S}_2})$ and $(\mathcal{S}_{1,2},\mathcal{E}_{1,2})$ satisfy outcome sharpness of MDEs and $(\mathcal{S}_1,\mathcal{E}_1)$ satisfies uniform superposition, then $(\mathcal{S}_{1,2},\mathcal{E}_{\mathcal{S}_{1,2}})$ satisfies uniform  superposition.
    \end{enumerate}
    \label{Thm::SuperpositionSubsystems}
\end{restatable}

On the other hand, mutual superposition is an exception and does not always carry over to arbitrary compositions. This observation will lead to our main result formalised in Corollary~\ref{corollary:mainresult} to Theorem~\ref{Thm::QuantumTensorProduct}.

Note from Lemma~\ref{Thm::SuperpositionSubsystems} above that only one of the subsystems needs to admit superposition in order for the composition to admit the same. In fact, the subsystem that does not admit superposition can indeed be classical. This is an important distinction to~\cite{aubrun_entanglement_2022}. Therein, the authors first take superposition to be the non-simplicial structure of a state space. Their main result then says that if both the subsystems have non-simplicial state spaces, every non-signalling composition would then either have entangled states or entangled effects. Non-simpliciality of both state spaces is not necessary for non-simpliciality of the composite. For instance a gbit (non-simplicial) composed with a simplex will have a non-simplicial state space and therefore admit superposition in the sense of~\cite{aubrun_entanglement_2022}. However, there cannot be entanglement in either the state space or effect space since when one subsystem is a simplex, the minimal and maximal tensor products coincide. Lemma~\ref{Thm::SuperpositionSubsystems} only aks whether the composite will display superposition if one the subsystems do.

Converse to the question leading to Lemma~\ref{Thm::SuperpositionSubsystems}, one might also ask whether the presence of superposition in a composite system implies the presence of superposition in its subsystems. First, note that superposition exists in the bipartite state space of Boxworld, even though it does not in its subsystems.

\begin{example}
    The PR box can be represented as a bipartite distribution given by $\overrightarrow p_{\mathrm{PR}} \coloneqq (1/2)\delta_{a \oplus b = xy}$. Now take the collection of local deterministic distributions $\{\delta_{(a,0),(b,0)},\delta_{(a,x),(b,\overline b)},\delta_{(a,\overline x),(b,y)},\delta_{(a, \overline x),(b,1)}\}$~\footnote{\textit{Note:} All four of these states lie on the same Bell facet that is algebraically violated by the PR box. One can also construct an example where this is not the case.}. This collection can be perfectly discriminated by a set of extremal effects $\{e_i\}_{i=1}^{4}$, such that for any bipartite distribution, $\overrightarrow {p}$, $e_1(\overrightarrow {p})=p(00|00)$, $e_2(\overrightarrow {p})=p(01|00)$, $e_3(\overrightarrow {p})=p(10|01)$, and $e_4(\overrightarrow {p})=p(11|01)$. This set of effects forms a measurement because for any non-signalling distribution, the sum of the above probabilities is 1. In addition, the measurement is also an MDE since the maximum number of extremal states that can be perfectly discriminated is 4, and therefore the collection of states above is an MPD set. Now, notice that $e_1(\overrightarrow p_{\mathrm{PR}}) = e_3(\overrightarrow p_{\mathrm{PR}}) = 1/2$. Therefore, the PR box is a superposition of the the subset $\{\delta_{(a,0),(b,0)},\delta_{(a,x),(b,\overline b)}\}$ of the MPD set above. This is analogous to how the maximally entagled states $\ket{\phi_+}\rangle$, where $\ket{\phi_+} \coloneqq (\ket{00}+\ket{11})/\sqrt{2}$ is a superposition of the subset $\{\ket{00}\rangle,\ket{11}\rangle\}$ of the MPD set $\{\ket{00}\rangle,\ket{01}\rangle,\{\ket{10}\rangle,\ket{11}\rangle\}$.  
   \label{Example::prbox}
\end{example}

This example suggests that superposition can emerge from the nature of composition. In particular, whether entanglement implies the presence of superposition in the composite system, regardless of its presence or absence in the subsystems. For any theory with the property that the operational dimension of the composite system is equal to the product of the operational dimensions of its subsystems, this is indeed the case.

\begin{restatable}{lemma}{EntaglementImpliesSuperposition}
    Let $(\mathcal{S}_1, \mathcal{E}_{\mathcal{S}_1})$ and $(\mathcal{S}_2,\mathcal{E}_{\mathcal{S}_2})$ be two GPT systems with operational dimensions $d_1$ and $d_2$ respectively. Let $(\mathcal{S}_{1,2} , \mathcal{E}_{\mathcal{S}_{1,2}}) $ be their composition such that $\mathcal{S}_{1,2}  \supset \mathcal{S}_{1} \ntens \mathcal{S}_{2}$ and its operational dimension is $d_1d_2$. Then $(\mathcal{S}_{1,2}, \mathcal{E}_{1,2})$ exhibits superposition.   \label{Thm::EntanglementAndSuperposition}
\end{restatable}

Lemma~\ref{Thm::EntanglementAndSuperposition} establishes that, for the considered class of GPTs, the existence of entanglement---ensured by choosing a tensor product strictly larger than the minimal---implies the presence of superposition. In Boxworld, superposition is determined exclusively by entanglement, as implied by GLT not exhibiting superposition. The converse, i.e.\ whether superposition implies entanglement, is not true. An example is the minimal tensor product of quantum systems where the product state $\ket{++}\rangle$ is in superposition of the states $ \{ \ket{00}\rangle, \ket{01}\rangle, \ket{10}\rangle, \ket{11}\rangle \}$.

\section{Maximality of the Quantum Tensor Product}\label{Sec::MaximalityofQuantumTensorProduct}

The study of GPTs has been principally motivated by the goal of identifying reasonable principles to re-axiomatise quantum theory~\cite{Hardy2001FiveReasonableAxioms}. Information theoretic principles have been proposed since that try to single out the set of quantum correlations~\cite{PhysRevLett.96.250401,PhysRevLett.99.180502,Pawlowski2009,Navascues2009,Fritz2013,M_ller_2021}. Motivated by these, we investigate whether any of the superposition principles stated in Definition~\ref{Def::SuperpositionPrinciples} deem useful for this purpose. Here, we do not explore the connection between superposition and the set of correlations realisable in a causal structure~\footnote{A causal structure is a collection of variables arranged as nodes of a directed acyclic graph where some of the nodes are labelled as observed. It represents the causal relations among the variables.}. Rather, we investigate, whether a superposition principle can replace one of the textbook axioms of quantum theory.

Recall that mutual superposition is missing from Lemma~\ref{Thm::EntanglementAndSuperposition}. This came from the observation that the maximal tensor product of quantum systems does not admit mutual superposition, even though the subsystems do. Since the superposition principles were inspired from the structure of quantum state and effect spaces, the quantum tensor product admits mutual superposition. It is then natural to ask whether there are compositions of quantum systems, strictly larger than the quantum tensor product, that admit mutual superposition. We found that if one also assumes no-restriction hypothesis, no strict extension will admit mutual superposition, as formalised below.

\begin{restatable}{theorem}{QuantumTensorProduct}
    Let $(\mathcal{S_Q}_1, \mathcal{E}_{\mathcal{S_Q}_1})$ and $(\mathcal{S_Q}_2, \mathcal{E}_{\mathcal{S_Q}_2})$ be the state and effect spaces describing two quantum systems respectively. Let $(\mathcal{S_Q}_{1,2}, \mathcal{E}_{\mathcal{S_Q}_{,12}})$ be the quantum state and effect spaces of the composite system. No composition $(\mathcal{S}_{\mathcal{Q}_1} \boxtimes \mathcal{S}_{\mathcal{Q}_2},\mathcal{E}_{\mathcal{S}_{\mathcal{Q}_1} \boxtimes \mathcal{S}_{\mathcal{Q}_2}})$ where $\mathcal{S_Q}_{1,2}  \subset   \mathcal{S}_{\mathcal{Q}_1} \boxtimes \mathcal{S}_{\mathcal{Q}_2} \subseteq \mathcal{S}_{\mathcal{Q}_1} \otimes_{\text{max}} \mathcal{S}_{\mathcal{Q}_2}$ admits mutual superposition.

    \label{Thm::QuantumTensorProduct}
\end{restatable}

This result provides an operational characterisation of the quantum composition rule, offering an alternative to the requirement of textbook quantum tensor product. On the other hand, recall from Lemma~\ref{Thm::SuperpositionSubsystems} above, the maximal tensor product of quantum systems satisfies both complete and uniform superposition. Together with~\ref{Thm::QuantumTensorProduct}, we then get the following consequence.

\begin{corollary}
    No strict extension of the quantum tensor product simultaneously admits all three superposition principles and the no-restriction hypothesis.
    \label{corollary:mainresult}
\end{corollary}

Although the statements in Theorem~\ref{Thm::QuantumTensorProduct} and Corollary~\ref{corollary:mainresult} deal with bipartite systems, one can consider two systems at a time to extended this result to multipartite systems.

\section{Preparational Uncertainty} \label{Sec::RelationToUncertainty}

Preparational uncertainty and superposition are usually studied as separate non-classical features of quantum theory. Since an operational notion of superposition is available, we investigate how uncertainty connects to superposition. We outline a brief introduction to preparational uncertainty and show that preparational uncertainty is a stricter non-classical requirement than superposition.  

Consider a scenario, where Alice prepares a state $s$ and sends it to Bob. Bob is allowed to randomly perform one of two measurements $\mathsf{M}$ or $\mathsf{N}$ determined by a random variable $X$ taking values $x \in \{0,1\}$ depending on the choice of measurements. His observation, labelled by a random variable $A$ taking values $a \in \{0,\cdots,d-1\}$, is reported back to Alice. Alice needs to guess Bob's outcome. Quantum theory allows for MDE measurements $\mathsf{M}$ and $\mathsf{N}$, such that for no choice of prepared state, Alice can perfectly guess, on average. In particular, the distribution of Alice's guess admits:
\begin{equation}
    \begin{split}
        &\mathrm{P}(A) = p(x=0)\mathrm{P}(A|x=0) + p(x=1)\mathrm{P}(A|x=1) < 1 \\
        & \implies \mathrm{H}_{\text{min}}(A) \coloneqq -\log_2 \left(\mathrm{P}(A)\right) > 0 \, , \\
    \end{split}
\end{equation}
where $\mathrm{H}_{\text{min}}(A)$ is the min-entropy of Alice's guessing probability and the implication is known as an \textit{entropic uncertainty relation}~\cite{maassen_generalized_1988}. This phenomenon is known as \textit{preparational uncertainty}\footnote{Preparational uncertainty \cite{heisenberg_uber_1927,robertson_uncertainty_1929, maassen_generalized_1988, fiorentino2022uncertainty} is not to be confused with measurement uncertainty \cite{busch2014measurement, busch2015direct}, which pertains to the impossibility to measure two observables jointly. Both forms of uncertainty have been discussed extensively in the literature in the context of operational theories \cite{busch_comparing_2013, takakura_preparation_2020, saha_operational_2020, designolle_incompatibility_2019, aravinda_origin_2017}. Here, we focus on preparational rather than measurement uncertainty, as the former admits a natural formulation in terms of properties of states (see Definition~\ref{Def::PrepUncertainty}), making it particularly suited to drawing a parallel with our notion of superposition.}. The essence of this relies on the fact that there are no states that are deterministic with respect to both $\mathsf{M}$ and $\mathsf{N}$. This can be phrased theory-independently as follows:

\begin{definition}[Preparational Uncertainty]
    Two MDEs $\mathsf{M}$ and $\mathsf{N}$ exhibit preparational uncertainty iff $\Det(\mathsf{M}) \cap \Det(\mathsf{N}) = \emptyset$.
    \label{Def::PrepUncertainty}
\end{definition} 
Definition~\ref{Def::PrepUncertainty} introduces preparational uncertainty as a relational property for pairs of MDEs. A GPT system will then exhibit preparational uncertainty if it contains at least a pair of MDE with such property. 

Similarly to our treatment of superposition, we can introduce principles that generalise different structural features found in quantum theory that are related to preparational uncertainty. \textit{Weak uncertainty} denotes the property that, for every non-degenerate PVM of an arbitrary quantum system, one can find another non-degenerate PVM with no shared eigenstates. \textit{Strong uncertainty} abstracts the exclusive property of qubit quantum theory where any pair of non-degenerate PVMs shares no common eigenstates.

\begin{definition}[Existence and Principles of Preparational Uncertainty]
    Let $(\mathcal{S},\mathcal{E})$ be a state and effect pair of a GPT. Preparational uncertainty \textbf{exists} iff some pair of MDEs $\mathsf{M},\mathsf{N} \subseteq \mathrm{Ext}[\mathcal{E}]$  exhibits preparational uncertainty, following Definition~\ref{Def::PrepUncertainty}. Further, this uncertainty is said to be
     \begin{itemize}
        \item \textbf{weak} iff for every MDE $\mathsf{M}$ there exists an MDE $\mathsf{N}$ such that $\mathsf{M}$ and $\mathsf{N}$ exhibit preparational uncertainty;
        \item \textbf{strong} iff every pair of distinct MDEs $\mathsf{M},\mathsf{N} \subseteq \mathrm{Ext}[\mathcal{E}]$ exhibits preparational uncertainty.
    \end{itemize}
    \label{Def::ExistencePrepUncertainty}
\end{definition}

Table \ref{Tab::UncertaintyPrinciples} summarises these properties for the GPTs considered; detailed derivations are provided in Appendix \ref{Appendix: superposition GPT examples}.

% \begin{table}[ht]
%     \centering
%     \centering
%     \renewcommand{\arraystretch}{1.3} 
%     \begin{tabularx}{\columnwidth}{@{} l c YY @{}}
%         \toprule
%         \multirowthead{2}{\textbf{GPT}} & 
%         \multirowthead{2}{\textbf{Existence of} \\ \textbf{Uncertainty}} & 
%         \multicolumn{2}{c}{\textbf{Uncertainty Principles}} \\
%         \cmidrule(lr){3-4}
%         & & weak & strong \\
%         \midrule
%         Quantum ($d =2$) & \checkmark & \checkmark & \checkmark \\
%         Quantum ($d >2$) & \checkmark & \checkmark & $\times$ \\
%         Classical & $\times$ & \textbf{--} & \textbf{--} \\
%         GLT   & $\times$ & \textbf{--} & \textbf{--} \\ 
%         Boxworld & $\times$ & \textbf{--} & \textbf{--} \\ \thinhline
%         Spekkens  & \checkmark & \checkmark & \checkmark \\
%         $n$-gon ($n \leq 6$) & $\times$ & \textbf{--} & \textbf{--} \\
%         $n$-gon ($n>6$)      & \checkmark & \checkmark & $\times$ \\
%         GPT-1  & $\times$ & \textbf{--} & \textbf{--} \\
%         GPT-2 & $\times$ & \textbf{--} & \textbf{--} \\
%         \bottomrule
%     \end{tabularx}
%     \caption{Classification of preparational uncertainty in various GPTs. Theories above the horizontal line include both single and composite systems, whereas for those below the line we only examine individual systems. GPT-1 and GPT-2 are defined in Appendix \ref{Appendix::GPT1_2}. The symbol $\textbf{--}$ indicates that the principles are not applicable.}
%     \label{Tab::UncertaintyPrinciples}
% \end{table}

%NEW TABLE
\begin{table}[ht]
    \centering
    \includegraphics[width=1.0\linewidth]{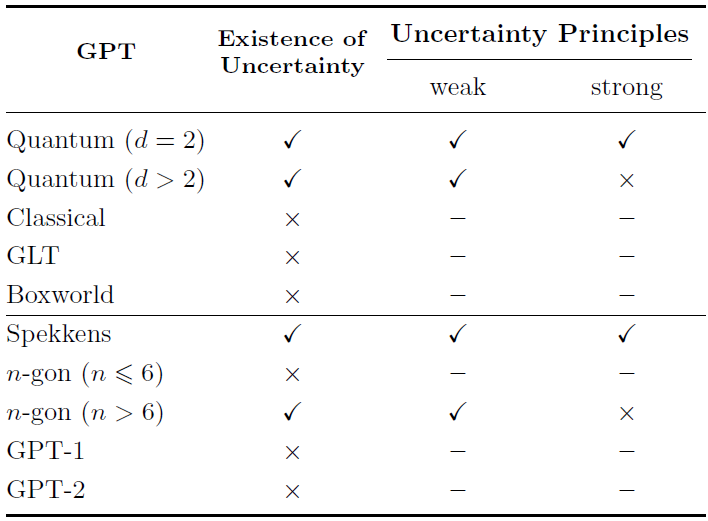}
    \caption{Classification of preparational uncertainty in various GPTs. Theories above the horizontal line include both single and composite systems, whereas for those below the line we only examine individual systems. GPT-1 and GPT-2 are defined in Appendix \ref{Appendix::GPT1_2}. The symbol $\textbf{--}$ indicates that the principles are not applicable.}
    \label{Tab::UncertaintyPrinciples}
\end{table}

A comparison with Table \ref{Tab::SuperpositionPrinciples} illustrates that superposition and preparational uncertainty are not operationally equivalent. For instance, we mentioned that strong uncertainty only holds for individual qubits ($d=2$), where each projector belongs to a unique PVM. Furthermore, in contrast to the superposition principles, the uncertainty principles defined here form a hierarchy: strong uncertainty necessarily implies weak uncertainty.

Preparational uncertainty is, in fact, a strictly stronger notion than superposition. Lemma~\ref{Thm::UncertaintyImpliesSuperposition} below formalises this by establishing relationships of implication between some of the principles in Definition~\ref{Def::SuperpositionPrinciples} and the existence of preparational uncertainty in Definition~\ref{Def::ExistencePrepUncertainty}. The last part of the proof relies on the additional operational assumption of outcome sharpness of MDEs.

\begin{restatable}{lemma}{UncImpliesSup}
    Let $ (\mathcal{S},\mathcal{E})$ be a state and effect space pair of a GPT such that preparational uncertainty exists. Then,
\begin{enumerate} [label=(\roman*), font=\normalfont]
    \item $(\mathcal{S},\mathcal{E})$ admits complete superposition   
    
    \item if $ (\mathcal{S},\mathcal{E})$ admits outcome sharpness of MDEs, it also admits uniform superposition.
\end{enumerate}
    \label{Thm::UncertaintyImpliesSuperposition}
\end{restatable}

The connection between superposition and uncertainty is made even more clear when outcome sharpness of MDEs is assumed from the outset. In such a class of GPTs---which include classical and quantum systems (for all $d$), Spekkens' model, and the toy theories GPT-1 and GPT-2---existence and weak preparational uncertainty can be naturally reformulated entirely in terms of superposition states, as follows (see Appendix \ref{Appendix: Restating uncertainty wrt superposition} for proofs):

\begin{restatable}{corollary}{UncertaintyAsSuperposition}
    Let $(\mathcal{S},\mathcal{E})$ be a state and effect space pair admitting outcome sharpness of MDEs. Then:
    
    \begin{enumerate} [label=(\roman*), font=\normalfont]
    \item existence of preparational uncertainty is equivalent to the existence of MPD sets, $\mathsf{D},\mathsf{D}'$, such that every $s \in \mathsf{D}$ is a superposition of $\mathsf{D}'$;
    \item weak uncertainty is equivalent to the condition that, for every MPD set $\mathsf{D}$, there exists an MPD set $\mathsf{D}'$ such that all $s \in \mathsf{D}$ are superpositions of $\mathsf{D}'$.
    
\end{enumerate}
\label{Lem::UncertaintyAsSuperposition}
\end{restatable}

An analogous formulation for strong uncertainty, that does not rely on additional assumptions, is unknown to us. It should be noted, however, that the requirement of strong uncertainty effectively singles out real qubit quantum theory within the class of polygonal theories. In these models, preparational uncertainty arises only beyond the threshold $n=6$, wherein weak uncertainty is satisfied; however, strong uncertainty only emerges in the quantum ($n \to \infty$) limit.

\section{Discussion} \label{Sec:Discussion}

We introduced an operational notion of superposition (Definition~\ref{Def::Superposition}) that is independent of the traditional Hilbert space formalism of quantum theory and that captures the essential statistical information needed in a prepare-measure scenario to identify the presence of superposition in a theory. Using this, we introduced three principles of superposition (Definition~\ref{Def::SuperpositionPrinciples}) to capture the structural properties of theories related to superposition, that are otherwise equivalent in quantum theory. As our main result (Theorem~\ref{Thm::QuantumTensorProduct} and Corollary~\ref{corollary:mainresult}), we manage to show that one of these principles---mutual superposition---singles out the quantum tensor product as the largest composition of quantum state spaces that admit it. Therefore, given that the subsystems are quantum, this principle serves as an operational statement as an alternate to the composition rule in the axioms of textbook quantum mechanics. 

We also identified conditions under which the presence of superposition in subsystems carries over to any of their non-signalling compositions (Lemma~\ref{Thm::SuperpositionSubsystems}). In addition, we found that the presence of entanglement implies the existence of superposition (Lemma~\ref{Thm::EntanglementAndSuperposition}), although the reverse implication is false. Finally, we found that preparational uncertainty, which is usually studied separately from superposition and entanglement, implies the existence of superposition as well (Lemma~\ref{Thm::UncertaintyImpliesSuperposition}). In fact, uncertainty can be seen as a special form of superposition (Corollary~\ref{Lem::UncertaintyAsSuperposition}).

There have been attempts at ruling out non-quantum tensor product compositions of quantum systems. However, to our knowledge, these principles are either not operational in nature, or they do not manage to single out the quantum tensor product~\cite{Weyl1931GroupsQuantumMechanics,AertsDaubechies1978TensorProduct,Wootters1990LocalAccessibility,Hardy2001FiveReasonableAxioms,DakicBrukner2011EntanglementSpecial,MasanesMueller2011PhysicalRequirements,ChiribellaDArianoPerinotti2011Informational,CarcassiMacconeAidala2021FourPostulates,SenEtAl2022TimelikeCorrelations,PatraEtAl2023InformationCausality,ErbaPerinotti2024CompositionRule}. 
In this regard, the principle of mutual superposition, identifying the quantum tensor product, provides a partial answer towards a reconstruction of quantum theory from physically motivated principles. However, we need to assume that the subsystems are quantum. Whether a set of operational conditions sufficient to identify the quantum state space for single systems can be derived remains an open problem. Another direction would be to extend our results to infinite dimensions.

Recently proposed tests for the quantumness of gravity~\cite{marletto_gravitationally_2017,PhysRevLett.119.240401} start with two masses that are \textit{held} in superposition of positions which are then subject to the gravitational field. After the evolution, if the masses are entangled, then one infers that the action of the gravitational field cannot be modelled by a channel that can be decomposed as local operations and classical communication. However, since these proposals are formulated within the Hilbert-space formalism, in particular with the example of the Mach-Zehnder interferometer, it is not so clear what it means operationally to \textit{hold} a mass in a superposition of positions. Here, one should refrain from interpreting a superposition of positions as implying that a mass is simultaneously localised at two distinct locations. Since mass generates a gravitational field, such an interpretation would already presuppose that gravity can be non-classical, which is at the end what one would like to show or falsify. A useful direction of research would be to use an operational notion of superposition to clearly identify which degrees of freedom are in superposition and would undergo entanglement in the protocols described.

\section{Acknowledgements}
The authors thank Stefan Weigert and Marco Erba for useful discussions. VF was supported by the Leverhulme Trust, grant number RPG-2024-201. KS l’Agence Nationale de la Recherche (ANR) project ANR-22-CE47-001.  

\bibliography{GPTs}

\newpage

\onecolumngrid
\appendix

\section{Superposition and Preparational Uncertainty in GPTs -- Examples} \label{Appendix: superposition GPT examples}

\subsection{Quantum and Classical theories}

\paragraph{Superposition}
We mentioned in the justification following Definition~\ref{Def::SuperpositionPrinciples} that, by construction, quantum theory satisfies all three superposition principles introduced in Definition~\ref{Def::SuperpositionPrinciples}. Specifically, the fact that for every (extremal) unit-trace positive semi-definite matrix, there exists a basis of the Hilbert space in which it is not diagonal implies that complete superposition is satisfied. Uniform superposition follows from the outcome sharpness of MDEs (see Sec.\ \ref{Sec::RelationToUncertainty}), according to which every MDE $\mathsf{M}$ of a quantum system (that is, every non-degenerate PVM) perfectly distinguishes a unique MPD set $\mathsf{D}$. Consequently, every state not belonging to $\mathsf{D}$ will be probabilistic with respect to $\mathsf{M}$. Finally, mutual superposition follows from the property that, if $\ket \psi = \sum_{i =1}^{d^{\prime}} c_i \ket\phi_i$ holds with $c_i \neq 0$ for all $i$, then there exists an orthonormal basis of $\mathcal{H}$ including $\ket\psi$, e.g.\ $\{ \ket\psi, \ket\psi_1, \dots, \ket\psi_{d-1} \}$, in which the operators $\ket{\Phi_i}\rangle$ are not diagonal. Existence of superposition and all three associated principles are preserved under the standard composition of quantum system, which ensures that the state (effect) space of a composite system with, say, two subsystems with operational dimensions $d_A$ and $d_B$, respectively, is isomorphic to the state (effect) space of a non-composite system with operational dimension $d_Ad_B$.

Classical theories do not exhibit superposition as per Definition~\ref{Def::Superposition}. They are characterised by simplicial state spaces, where the vertices (i.e.\ the set of all allowed extremal states) form the unique MPD set in the theory. Likewise, classical composite systems, for which the minimal and maximal tensor products coincide, feature a single product MDE.

\paragraph{Preparational Uncertainty}

In quantum theory, any pair of MDEs $\mathsf{M}=\{ P_i \}_{i=1}^d$ and $\mathsf{M}=\{ Q_i \}_{j=1}^d$ such that $[P_i,Q_j]\neq 0$ for all $i,j$ exhibit preparational uncertainty. In every Hilbert space $\mathcal{H}$ with finite $d \coloneqq \text{dim}(\mathcal{H})$, given any MDE $\mathsf{M}$, there exist infinitely many MDEs $\mathsf{N}$ such that $(\mathsf{M},\mathsf{N})$ exhibits uncertainty, hence weak uncertainty (Definition~\ref{Def::ExistencePrepUncertainty}) is satisfied for both single and composite systems. However, for $d\geq 3$, a given rank-1 projector $P_i$ belongs to infinitely many different MDEs. This implies that strong uncertainty does not apply to quantum systems, both single and composite, with the exception of a single qubit. In $d=2$, in fact, each $P_i$ belongs to a unique MDE: in the terminology of Gleason \cite{Gleason1975}, projective qubit measurements do not \textit{intertwine}. Classical theories do not exhibit preparational uncertainty as each single or composite system features a single MDE.

\subsection{GLT and Boxworld}

\paragraph{Superposition}
In GLT, the MDEs are given by collections of ray-extremal effects that sum up to the unit effect. Since the inner product between any extremal state and a ray-extremal effect is either 0 or 1, there does not exist an extremal state which is probabilistic with respect to an MDE measurement. Consequently, GLT does not admit any principle of superposition. On the other hand, in Boxworld, where the MDEs are also formed by collection of ray-extremal effects, there exist extremal states, namely PR-boxes, that are probabilistic. We will show this with an example. We represent a state by

\begin{equation}
\begin{pmatrix}
   p(0,0|0,0)  &   p(0,1|0,0) &   p(0,0|0,1) &  p(0,1|0,1)  \\
  p(1,0|0,0)   &   p(1,1|0,0) &   p(1,0|0,1) &  p(1,1|0,1)  \\ 
   p(0,0|1,0)   &  p(0,1|1,0)  &  p(0,0|1,1)  &   p(0,1|1,1) \\
   p(1,0|1,0)  &   p(1,1|1,0)  &  p(1,0|1,1)  &  p(1,1|1,1)  \\
\end{pmatrix}
\end{equation}
where $p(a,b|x,y)$ denotes the probability of getting the outcomes $a$ and $b$ when fiducial measurements $x$ and $y$ are performed. In this representation, the PR box is given as:
\begin{equation}
\frac{1}{2}\begin{pmatrix}
   1  &   0 &   1 &  0  \\
  0   &   1 &   0 &  1  \\ 
  1   &  0  & 0  &   1 \\
   0  &   1  &  1  &  0  \\
\end{pmatrix}
\end{equation}
and the MDE with respect to which it is probabilistic is 
\begin{equation}
\left\{ \begin{pmatrix}
   1  &   0 &   0 &  0  \\
  0   &   0 &   0 &  0  \\ 
  0   &  0  & 0  &   0 \\
   0  &   0  &  0  &  0  \\
\end{pmatrix}, \begin{pmatrix}
   0  &   1 &   0 &  0  \\
  0   &   0 &   0 &  0  \\ 
  0   &  0  & 0  &   0 \\
   0  &   0  &  0  &  0  \\
\end{pmatrix},\begin{pmatrix}
   0  &   0 &   0 &  0  \\
  1   &   0 &   0 &  0  \\ 
  0   &  0  & 0  &   0 \\
   0  &   0  &  0  &  0  \\
\end{pmatrix}, \begin{pmatrix}
   0  &   0 &   0 &  0  \\
  0   &   1 &   0 &  0  \\ 
  0   &  0  & 0  &   0 \\
   0  &   0  &  0  &  0  \\
\end{pmatrix} \right\}.
\end{equation}
This MDE perfectly discriminates the following collection of extremal states:
\begin{equation}
\left\{ \begin{pmatrix}
   1  &   0 &   1 &  0  \\
  0   &   0 &   0 &  0  \\ 
  1   &  0  & 1  &   0 \\
   0  &   0  &  0  &  0  \\
\end{pmatrix}, \begin{pmatrix}
   0  &   1 &   1 &  0  \\
  0   &   0 &   0 &  0  \\ 
  0   &  1  & 1  &   0 \\
   0  &   0  &  0  &  0  \\
\end{pmatrix},\begin{pmatrix}
   0  &   0 &   0 &  0  \\
  1   &   0 &   1 &  0  \\ 
  1   &  0  & 1  &   0 \\
   0  &   0  &  0  &  0  \\
\end{pmatrix}, \begin{pmatrix}
   0  &   0 &   0 &  0  \\
  0   &   1 &   1 &  0  \\ 
  0   &  1 & 1  &   0 \\
   0  &   0  &  0  &  0  \\
\end{pmatrix} \right\},
\end{equation}
which are all separable states. Therefore the PR box is a superposition of these separable states. This is analogous to the maximally entangled state $\ket{\Phi^+}\rangle$ being a superposition of the separable states $\ket{00}\rangle$ and $\ket{11}\rangle$. However, since there is no MDE that can discriminate four PR boxes, the stronger superposition principles of Definition~\ref{Def::SuperpositionPrinciples} are not admitted in Boxworld.

\paragraph{Preparational Uncertainty}
Following from the fact that the MDEs in both GLT and Boxworld can only perfectly discriminate extremal separable states, and that the inner product of such states with any ray extremal effect is either 0 or 1, we infer that neither of these theories admit existence of preparational uncertainty. Therefore, the reamaining principles are not admitted as well.

\subsection{Spekkens' toy model}

\paragraph{Superposition}

Spekkens’ toy theory \cite{spekkens_evidence_2007} was originally introduced to show that many phenomena of quantum theory can be explained via a ``classical'' epistemic restriction on underlying hidden variables. The model is functionally equivalent to the stabilizer formalism of a single qubit \cite{pusey2012stabilizer}, and comprises six extremal (epistemic) states (see Fig.\ \ref{Fig::SpekkensOctahedron}) and three MDE measurements (analogous to Pauli $X,Y,Z$ measurements). Each MDE distinguishes a unique pair of extremal states.

\begin{figure}
    \centering
    \includegraphics[width=0.3\linewidth]{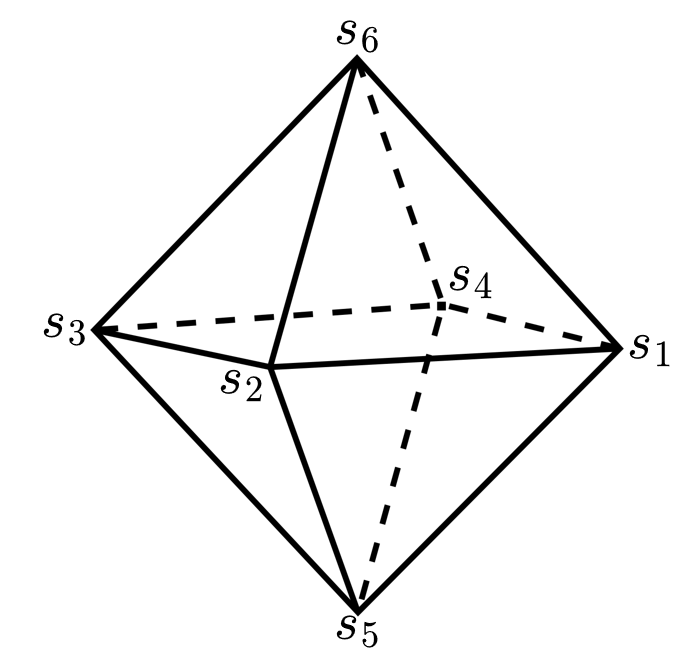}
    \caption{Sketch of the state space of Spekkens' toy model, with the six extremal states labeled.}
    \label{Fig::SpekkensOctahedron}
\end{figure}

By construction, for any extremal state $s$ that does \textit{not} belong to the unique MPD set associated with a measurement $\mathsf{M}=\{e_1, e_2\}$, the model satisfies the condition $e_i(s)=1/2$ for $i=1,2$. This defining characteristic, known as the \textit{equivalence-balance principle}, ensures that the epistemic states not distinguished by a measurement appear maximally uncertain with respect to its outcomes. This principle directly implies the existence of superposition within the theory and further guarantees that the model satisfies all three superposition principles in Definition~\ref{Def::SuperpositionPrinciples}.

\paragraph{Preparational Uncertainty}

The equivalence-balance principle ensures that all three pairs of MDEs allowed in the theory exhibit uncertainty (strong uncertainty). It follows that, along with qubit quantum theory, Spekkens' model is another example of GPT that satisfies all uncertainty principles of Definition~\ref{Def::ExistencePrepUncertainty}.

\subsection{Regular $n$-gons} \label{Appendix:RegularNgons}

\paragraph{Superposition}
For $n > 4$, let $\mathcal{S}_n$ be the state space formed by the convex hull of the extremal points
\begin{equation}
\omega_n(i)=
  \begin{pmatrix}
    r_n\cos\frac{2\pi i}{n}\\
    r_n\sin\frac{2\pi i}{n}\\
    1
  \end{pmatrix},
\end{equation}
where 
\begin{equation}
r_n=\sqrt{\frac{1}{\cos(\pi/n)}},
\end{equation}
which define the vertices of the regular $n$-gon \cite{Janotta_2011, janotta2013generalized, massar2014information}. The ray-extremal effects in the effect space $\mathcal{E}_n$ are 
\begin{equation}
e_n(i)=
  \begin{cases}
    \displaystyle
    \frac12
    \begin{pmatrix}
      r_n\cos\frac{(2i+1)\pi}{n}\\
      r_n\sin\frac{(2i+1)\pi}{n}\\
      1
    \end{pmatrix}, & n\text{ even},\\[14pt]
    \displaystyle
    \frac{1}{1+r_n^{2}}
    \begin{pmatrix}
      r_n\cos\frac{2\pi i}{n}\\
      r_n\sin\frac{2\pi i}{n}\\
      1
    \end{pmatrix}, & n\text{ odd},
  \end{cases}
\end{equation}
with the unit effect being $u=(0,0,1)^{\!\top}$. When $n$ is odd, the complements $\{u-e_n(i)\}_i$ are also extremal effects (but not ray-extremal), and the corresponding GPT $(\mathcal{S}_n,\mathcal{E}_n)$ is (strongly) self-dual.
To simplify the expressions, we introduce the planar angles for effects and states, respectively:
\begin{equation}
    \theta_j^n = \begin{cases}
        \frac{(2j+1)\pi}{n}, & n \text{ even} \\
        \frac{2\pi j}{n}, & n \text{ odd}
    \end{cases} \,, \quad \quad \quad \varphi_i^n = \frac{2 \pi i}{n} \, .
\end{equation}
Then the inner product takes the form
\begin{equation}\label{eq:eval}
  e_n(j)\bigl(\omega_n(i)\bigr)=
  \frac{1+r_n^{2}\cos(\vartheta_j^n-\varphi_i^n)}{D_n},
\end{equation}
where
\begin{equation}
D_n=
  \begin{cases}
    2, & n\text{ even},\\
    1+r_n^{2}, & n\text{ odd}.
  \end{cases} 
\end{equation}

For all $n$, the operational dimension is two. In particular, for odd $n$, one has $ e_n(j)\bigl(\omega_n(i)\bigr)=1$ when $\theta_j^n-\varphi_i^n=0$---that is, when $i=j$. This means that, for each extremal state $\omega_n(i)$, there exists a unique ray-extremal effect, namely $e_n(i)$, for which it gives probability $1$ (while giving probability strictly less than $1$ for any other state). 
 
In addition, the equation $ e_n(j)\bigl(\omega_n(i)\bigr)=0$ is satisfied when $\cos(\vartheta_j^n-\varphi_i^n) = - \cos(\pi/n)$, i.e.\ when $\theta_j^n-\varphi_i^n=\pi \pm \pi/n $ (mod $2\pi$). Therefore, for each $\omega_n(i)$, there are two ray-extremal effects $e_n(j)$, where $j=i + \frac{(n \pm 1)}{2}$ (mod $n$), assigning it probability zero. For example, both $\{\omega_7(2), \omega_7(6)\}$ and $\{\omega_7(2), \omega_7(5)\}$ are MPD sets that are discriminated by the MDE measurement $\{e_7(2), u-e_7(2) \}$.

Repeating the calculation for even $n$, we find that $ e_n(j)\bigl(\omega_n(i)\bigr)=1$ when $\theta_j^n-\varphi_i^n= \pm \pi/n$. Fixing $i$, the two solutions are $j=i$ and $j=i-1$ (mod $n$).

Similarly, $ e_n(j)\bigl(\omega_n(i)\bigr)=0$ is satisfied when $j=i+n/2$ or when $j=i+(n-2)/2$ (mod $n$). It follows that the same MDE measurement $\{e_n(j), u-e_n(j)\}$ perfectly distinguishes a total of \textit{four} MPD sets of extremal states.

Fig.\ \ref{Fig::PentagonHexagon} depicts the state spaces of the pentagon ($n=5$) and the hexagon ($n=6$) models, highlighting one MDE and the maximal sets associated with it.

\begin{figure}
    \centering
    \includegraphics[width=0.6\linewidth]{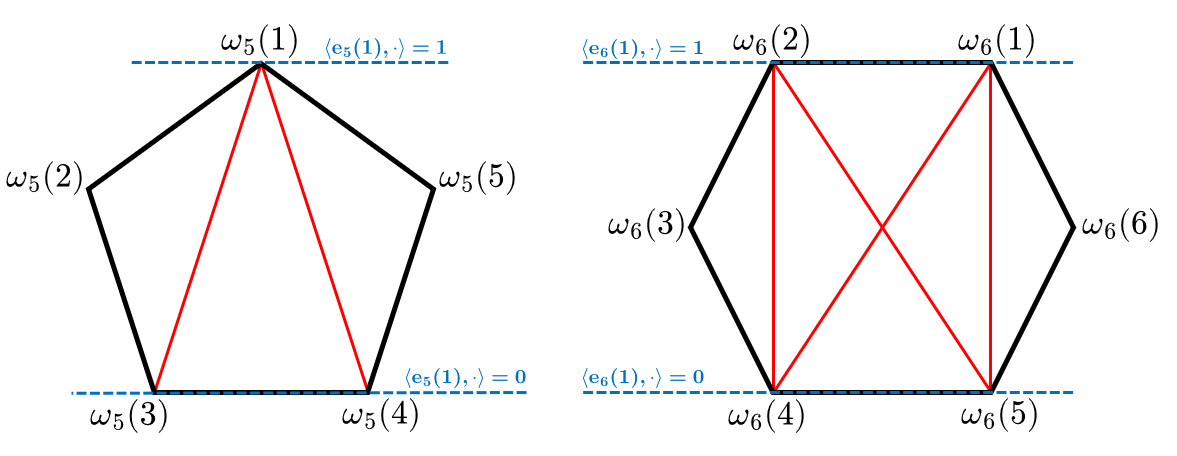}
    \caption{The pentagon ($n=5$; left) and hexagon ($n=6$; right) state spaces. Pairs of vertices connected by red segments constitute the MPD sets distinguished by the MDE $\mathsf{M}=\{ e(1), u-e(1)\}$. Points lying on the top (resp.\ botom) blue line yield probability $1$ (resp.\ $0$) for the effect $e_1$.}
    \label{Fig::PentagonHexagon}
\end{figure}

For any odd $n$, a given MDE $\mathsf{M} = \{e_n(j), u-e_n(j)\}$ has a total of three deterministic states, $\Det(\mathsf{M})=\{\omega_n(j), \omega_n(j-(n+1)/2), \omega_n(j-(n-1)/2)\}$ (all sums of indices are mod $n$). Graphically, MPD sets always corresponds to pairs of vertices where one is an endpoint of the segment opposite the other vertex. It follows that these polygonal GPTs exhibit superposition, since they all count more than three vertices. In particular, complete superposition holds: every extremal state is probabilistic with respect to some MDE. Uniform superposition also applies: each MPD set $\mathsf{D}$ is distinguished by exactly two MDEs, and every extremal state not in $\mathsf{D}$ is probabilistic to at least one of them (Fig. \ref{Fig::PentagonSuperposition}). Moreover, given any state $r$ in a superposition of some MPD set $\mathsf{D}$, it is possible to find a extremal state $t$ at one end of the opposing side of the polygon, such that \textit{both} elements of $\mathsf{D}$ are probabilistic with respect to at least one of the two MDEs that distinguish $\{ r,t \}$. Therefore, mutual superposition applies to every regular $n$-gon with odd $n$.

\begin{figure}
    \centering
    \includegraphics[width=0.3\linewidth]{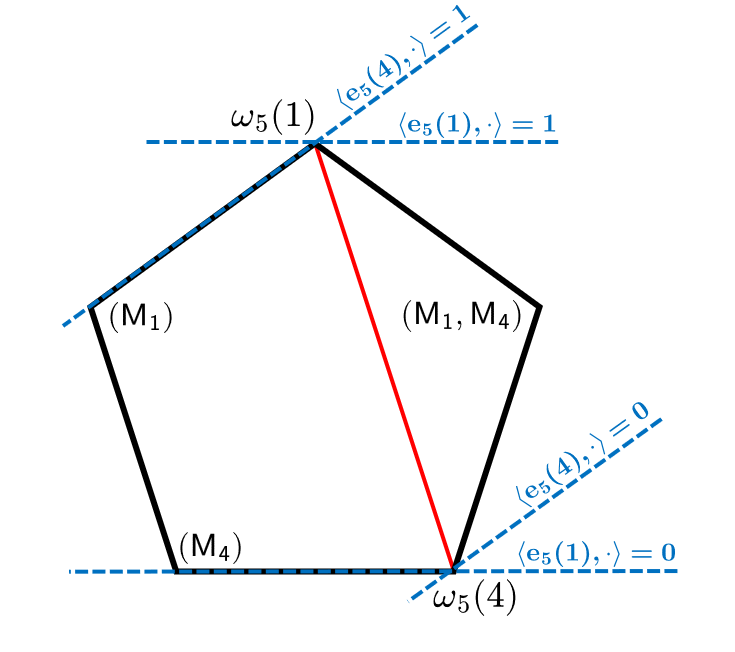}
    \caption{State space of the regular pentagon ($n=5$). The MPD set $\{\omega_5(1), \omega_5(4)\}$ is highlighted; vertices lying on the blue dotted lines correspond to the deterministic extremal states of either the MDEs $\mathsf{M}_1=\{ e_5(1), u-e_5(1)\}$ or $\mathsf{M}_4=\{ e_5(4), u-e_5(4)\}$. For each vertex not belonging to the MPD set, we indicate the MDE(s) for which they are probabilistic. Those that are probabilistic with respect to a single MDE are the extremal states that, while in a superposition of $\{\omega_5(1), \omega_5(4)\}$, form an MPD set with one of its elements.}
    \label{Fig::PentagonSuperposition}
\end{figure}

For even $n$, a given MDE $\mathsf{M} = \{e_n(j), u-e_n(j)\}  =\{ e_n(j), e_n (j+n/2)\}$ has four deterministic states: $\Det(\mathsf{M})=\{\omega_n(j), \omega_n(j+1), \omega_n(j-n/2), \omega_n(j-(n-2)/2)\}$ (all sums of indices are mod $n$). Superposition exists in these models since $n>4$; in particular, complete superposition applies because of the symmetric properties of the state and effect spaces: as in the odd case, each extremal state is probabilistic with respect to some MDE. In contrast with the odd case, however, here are two classes of MPD sets. In fact, any vertex forms an MPD pair with its antipodal vertex, but also with either of the two vertices adjacent to the antipodal vertex. In the first case (where the vertices are antipodal), there are two MDEs that distinguish them; whereas in the second case (where there vertices are not antipodal) there is only one MDE that distinguishes them (Fig.\ \ref{Fig::HexagonSuperposition}). In the first case, we find that any third extremal state is probabilistic with at least one of these MDEs---this is the same argument that justifies uniform superposition for odd $n$. However, for MPD pairs of the second kind, there will always exist two additional extremal states that are deterministic with respect to the unique MDE associated with it. This is true for \textit{any} even $n>4$: we conclude that, in general, uniform superposition does not hold in these models. A separate treatment for the two kinds of MPD sets is required to check whether mutual superposition holds. For a hexagon ($n=6$), we can provide an invalidating example: consider an arbitrary MPD set $\{ s_1, s_2 \}$ of the second type; the vertex between these two extremal states, call it $s_3$, will be probabilistic with respect to the unique MDE associated with $\{ s_1, s_2 \}$. Hence, it is a superposition of the two states. However, whatever MPD set we consider that includes $s_3$, all of their associated MDEs will include either $s_1$ or $s_2$ among their deterministic states (Fig.\ \ref{Fig::HexagonSuperposition}). Thus, mutual superposition is violated for $n=6$. The same is not true for larger even number of sides. In Fig.\ \ref{Fig::OctagonSuperposition}, we show this visually for the case of $n=8$: for either kind of MPD set, the superposition states (in blue) can be joined to form MPD sets with respect to which both of the original states are superpositions. 

\begin{figure}
    \centering
    \includegraphics[width=0.5\linewidth]{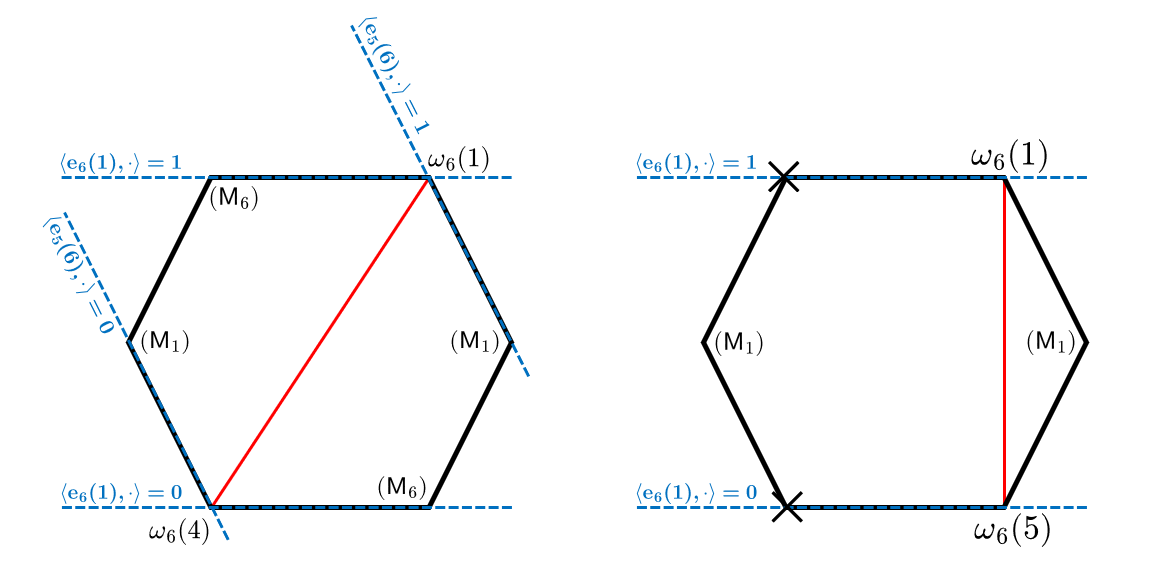}
    \caption{Regular hexagon GPT. We apply the same analysis as in Fig.\ \ref{Fig::PentagonSuperposition} to the MPD sets $\{\omega(1), \omega(4)\}$ (left; antipodal vertices) and $\{\omega(1), \omega(6)\}$ (right; non-antipodal vertices). Left: all other extremal states states that are probabilistic with respect to only one MDE between $\mathsf{M}_1=\{e_6(1), u-e_6(1)\}$ and $\mathsf{M}_6=\{e_6(6), u-e_6(6)\}$ hence, while in a superposition of $\{\omega(1), \omega(4)\}$, they each form a MPD set with one of its elements. Right: the existence of two vertices (distinguihsed by crosses) that are not probabilistic with respect to $\mathsf{M}_1$ demonstrates that the hexagon fails to satisfy uniform superposition.}
    \label{Fig::HexagonSuperposition}
\end{figure}

\begin{figure}
    \centering
    \includegraphics[width=0.4\linewidth]{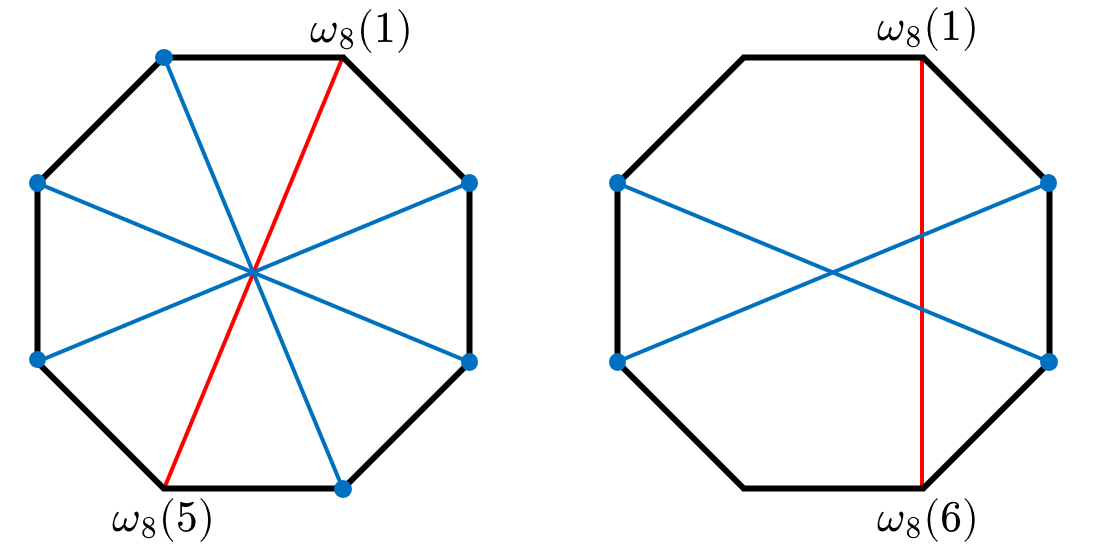}
    \caption{Regular octagon GPT ($n=8$): for each type of MPD set $\mathsf{D}$ (left: antipodal; right: non-antipodal), the states in a superposition of $\mathsf{D}$ (indicated by blue vertices) can all be connected (e.g., in the way shown by the blue segments) to form MPD pairs with respect to which both elements of $\mathsf{D}$ are a superposition of.}
    \label{Fig::OctagonSuperposition}
\end{figure}

As noted in the main text, regular $n$-gon GPTs exhibit a peculiar feature of superposition, which can be extrapolated from Figs.\ \ref{Fig::PentagonSuperposition}-\ref{Fig::HexagonSuperposition}. Regardless of the value of $n$ and the kind of MPD set considered, one can always identify a state $r$ that is a superposition of a pair $\{s_1,s_2\}$, yet simultaneously forms a valid MPD set with $s_1$ or with $s_2$. Such property hinges on the fact that a single MDE in a regular $n$-gon GPT perfectly distinguishes multiple MPD pairs---i.e.\ these models do not satisfy outcome sharpness of MDEs (see Sec.\ \ref{Sec::RelationToUncertainty}), in contrast with qubit quantum theory, where outcome sharpness of MDEs holds and this phenomenon is absent.

\paragraph{Preparational Uncertainty}
For both even and odd $n$, the fact that every extremal state is deterministic with respect to multiple MDEs implies that none of the GPTs in this class satisfy strong preparational uncertainty. Specifically, preparational uncertainty does not exist for $n=5,6$: for these regular polygons, the sets of deterministic states of any two MDEs always overlap (see Figure \ref{Fig::PentagonHexagon}). As the number of sides increases, such pairs do exist (cf.\ the octagon in Figure \ref{Fig::OctagonSuperposition}); in particular, due to the $n$-fold rotational symmetry of the $n$-gons, each MDE will belong to at least a pair of uncertain MDEs, i.e.\ weak preparational uncertainty applies.

\subsection{GPT-1 and GPT-2}
\label{Appendix::GPT1_2}

\paragraph{Superposition}
We call GPT-1 the theory whose state space $\mathcal{S}$ is a subset of the state space of a qutrit ($d=3$) and given by the
convex hull of $\ket{0}\rangle$ and the states $\ket{\psi}\rangle$, where 
\begin{equation}
\ket{\psi}=a\ket{1}+b\ket{2} : \alpha,\beta \in \mathbb{C}, |a|^2 + |b|^2=1 \, .
\end{equation}

The effect space $\mathcal{E}$ is the convex hull of the rank-1 qutrit projectors $\ket{0}\rangle$ and $\ket{\psi}\rangle$, the unit effect $u=\mathbb{I}_3$, and the zero operator. The MDEs in GPT-1 take the form
\begin{equation}
    \mathsf{M}_{\psi} = \{\ket{0}\rangle, \ket{\psi}\rangle, \ket{\psi^{\perp}}\rangle\} \, .
\end{equation}

Each $\ket{\psi}\rangle \neq \ket{0}\rangle$ is deterministic with respect to a single MDE, hence GPT-1 exhibits superposition as per Definition~\ref{Def::SuperpositionPrinciples}. However, by construction, the state $\ket{0}\rangle$ is deterministic with respect to all MDEs: it follows that complete superposition is not satisfied. At the same time, however, every state $\ket{\phi} \rangle \in \mathcal{S}$ such that $\ket{\phi}$ does not belong to the basis $\{ \ket{0}, \ket{\psi}, \ket{\psi^{\perp}} \}$, is probabilistic with respect to the MDE associated with it, hence uniform superposition holds. Moreover, as GPT-1 inherits outcome sharpness of MDEs from qutrit quantum theory, $\ket \phi \rangle$ identifies a unique MDE, namely $\{ \ket{0} \rangle, \ket{\phi}\rangle, \ket{\phi^{\perp}}\rangle \}$, with respect to which both $\ket{\psi}\rangle$ and $ \ket{\psi^{\perp}}\rangle $ are probabilistic: that is, mutual superposition is satisfied.

We call GPT-2 the theory whose state space is a subset of the state space of a qubit ($d=2$) and given by the convex hull of the three extremal qubit states $\ket 0 \rangle$, $\ket 1 \rangle$, and $\ket + \rangle$, where
\begin{equation}
    \ket{+}= \frac{\ket{1}+\ket{2}}{\sqrt{2}} \, .
\end{equation}
The extremal effects, as in GPT-1, are the corresponding rank-1 projectors $\ket 0 \rangle$, $\ket 1 \rangle$, $\ket + \rangle$, and those obtained by subtracting them to the unit effect $u=\mathbb{I}_2$. In GPT-2, there is a unique MPD set, $\mathsf{D} = \{\ket{0}\rangle, \ket{1}\rangle\}$, which is its own MDE, i.e.\ $\mathsf{M}=\mathsf{D}$. The only other extremal state, $\ket{+}\rangle$, is a superposition of $\mathsf{D}$ (i.e.\ uniform superposition applies). However, it does not belong to any MPD set (i.e.\ mutual superposition does not apply), which also means that neither element of $\mathsf{D}$ is a superposition (i.e.\ complete superposition does not apply).

\paragraph{Preparational Uncertainty}
In GPT-1, the extremal state $\ket{0}\rangle$ is deterministic with respect to all MDEs, hence the system does not exhibit preparational uncertainty. In GPT-2, there exist only one MPD, distinguished by a unique MDE, hence in this model there is also no preparational uncertainty.

\section{Proofs of Results}
\subsection{Superposition from Subsystems: Proof of Lemma~\ref{Thm::SuperpositionSubsystems}} \label{Appendix: superposition from subsystems}

\SuperpositionFromSubsystems*

\begin{proof}

    (i). If $(\mathcal{S}_{1}, \mathcal{E}_{1})$ exhibits superposition, then there exists some $s \in \text{Ext}[\mathcal{S}_1]$ and some MDE $\mathsf{M}_1=\{e_i\}$ such that $s \notin \Det(\mathsf{M}_1)$. Let $t$ be an arbitrary extremal element of $\mathcal{S}_{2}$ and $\mathsf{M}_2=\{f_j\}$ be an arbitrary MDE of system $2$, then $e_i \otimes f_j (s \otimes t) = e_i(s) f_j(t) <1$ for all $i$ and $j$, hence $s \otimes t \in \text{Ext}[\mathcal{S}_{1,2} ]$ is a superposition of (some of the elements of) any MPD set distinguished by $\mathsf{M}_1 \otimes \mathsf{M}_2$. Hence, superposition exists in the composite system $(\mathcal{S}_{1,2}, \mathcal{E}_{1,2})$.

    (ii). For any pair of MDEs $\mathsf{M}_1$ and $\mathsf{M}_2$ of systems $1$ and $2$, respectively, it follows that every composite state $z \in \text{Ext}[\mathcal{S}_{1,2}] \setminus \Det(\mathsf{M}_1 \otimes \mathsf{M}_2)$ is a superposition. Instead, for any $z \in \Det(\mathsf{M}_1 \otimes \mathsf{M}_2)$, there exists some $e_{i*} \otimes f_{j*}\in \mathsf{M}_1 \otimes \mathsf{M}_2$ such that $(e_{i*} \otimes f_{j*}) (z)=1$. Then $e_i(z_1) = e_i \otimes u_2 (z) = \sum_j (e_{i} \otimes f_j) (z) = \sum_j \delta_{ii*}\delta_{jj*} = \delta_{ii*}$. It follows that the marginal state $z_1$ is a (not necessarily extremal) state in $\mathcal{S}_{1}$ that is deterministic with respect to $\mathsf{M}_1$; in particular, $z_1 = \sum_k p_k s^{i*}_k$ where $\{p_k\}$ is a probability distribution, $s^{i*}_k \in \text{Ext}[\mathcal{S}_1]$ and $e_{i*}(s^{i*}_k)=1$ for all $k$. Consider some $k*$ such that $p_{k*} \neq 0$. Due to uniform superposition for subsystem $1$, it follows that for the state $s^{i*}_{k*}$ in the decomposition of $z_1$ there exists some MDE $\mathsf{N}_1$ such that $s^{i*}_{k*} \notin \Det(\mathsf{N}_1)$, i.e.\ $e'_l(s^{i*}_{k*}) <1$ for all $e'_l \in \mathsf{N}_1$. But this implies that $e'_l(z_1) = \sum_k p_k e'_l(s^{i*}_{k}) <1$ for all $l$, hence $z_1$ is not deterministic with respect to $\mathsf{N}_1$; in particular, $z_1$ cannot be written as a mixture of extremal states that give probability $1$ to \textit{the same} element of $\mathsf{N}_1$. Hence, for any choice of MDE $\mathsf{N}_2$ of subsystem $2$, we have that $z \notin \Det(\mathsf{N}_1\otimes \mathsf{N}_2)$. We conclude that for every element of $\Det(\mathsf{M}_1 \otimes \mathsf{M}_2)$ there must exist another product MDE $\mathsf{N}_1 \otimes \mathsf{N}_2$ for which it is probabilistic. Hence, every extremal element of the composite state space $\mathcal{S}_{1,2}$ is a superposition.

    (iii). This is an immediate consequence of the more general fact that, if outcome sharpness of MDEs holds in a GPT, then the existence of superposition alone implies uniform superposition. In fact, if each MDE $\mathsf{M}$ is associated with a unique MPD $\mathsf{D}_{\mathsf{M}}$, then every state $s \in \text{Ext}[\mathcal{S}]\setminus \mathsf{D}_{\mathsf{M}}$ will be probabilistic with respect to $\mathsf{M}$. Such extremal states exist if and only if superposition exists, i.e.\ when $\text{Ext}[\mathcal{S}] \supset  \mathsf{D}_{\mathsf{M}}$. In the context of the composite system $(\mathcal{S}_{1,2}, \mathcal{E}_{1,2})$, uniform (hence existence of) superposition in $(\mathcal{S}_{1}, \mathcal{E}_{1})$ implies existence of superposition in $(\mathcal{S}_{1,2}, \mathcal{E}_{1,2})$ (cf.\ point (i) above), and since $(\mathcal{S}_{1,2}, \mathcal{E}_{1,2})$ satisfies outcome sharpness of MDEs in the composite system, uniform superposition also applies.

\end{proof}

\subsection{Entanglement as a form of Superposition: Proof of Lemma~\ref{Thm::EntanglementAndSuperposition}} \label{Appendix: entanglement in superposition}

\EntaglementImpliesSuperposition*

\begin{proof}
    Let $z \in \text{Ext}[\mathcal{S}_{1,2} ]$ be entangled, and suppose the composite system does \textit{not} exhibit superposition. This means that $z$ is deterministic with respect to all MDEs of the composite system. In particular, we must have that $z \in \Det(\mathsf{M}_1 \otimes \mathsf{M}_2 )$ for all product MDEs $\mathsf{M}_1 \otimes \mathsf{M}_2$. We have $(e_i \otimes f_j) (z)=1$ for some $e_i \in \mathsf{M}_1$ and $f_j \in \mathsf{M}_2$ (and 0 for all other effects of $\mathsf{M}_1 \otimes \mathsf{M}_2$), hence $\sum_j (e_i \otimes f_j) (z)= (e_i \otimes \sum_j f_j)(z) =(e_i \otimes u_2)(z) \equiv e_i (z_1)=1 $, where $z_1$ is the (mixed) reduced state of $z$ for subsystem $1$ and $u_2$ is the unit effect of system $2$. This implies that $z_1$ can be written as a convex combination of maximal states giving probability $1$ to the local effect $e_i$: $z_1 = \sum_k p_k s_k$, where $p_k \neq 0$, $s_k \in \text{Ext}[\mathcal{S}_1]$ and $e_i(s_k)=1$ for all $k$. But since $z$ is deterministic with respect to all product MDEs, it follows that $z_1$ must be deterministic with respect to all MDEs of system $1$. That is, each extremal state $s_k$ (with the non-zero weight $p_k$) in the decomposition of $z_1$ must give probability one to \textit{exactly the same effect} in each MDE in $\mathcal{E}_1$. In other words, the states $s_k$ have the same (deterministic) probability distributions over all MDEs of system $1$. However, since MDEs form an informationally-complete set of measurements, it must follow that the $s_k$ states are \textit{the same state}, $s=s_k$ for all $k$. Hence $z_1 = s$ is a extremal state, in contradiction with the assumption that $z$ is entangled. We conclude that $z$ cannot be deterministic with respect to all product MDEs, hence it is a superposition of product states.
\end{proof}

\subsection{Singling out the Quantum Tensor Product: Proof of Theorem~\ref{Thm::QuantumTensorProduct}} \label{Appendix: quantum tensor product and superposition}

Some supporting lemmas.

\begin{lemma}
    If $\Pi_{1,2} = \ket{\phi}\rangle_{1,2}$ is a rank-1 bipartite projector and $\mathbb{I}_{1,2}$ denotes the identity matrix, then $\mathbb{I}_{1,2} - \Pi_{1,2}$ is separable. 
    \label{Lemma::separabilitu_comp_proj}
\end{lemma}
\begin{proof}
    In~\cite[Theorem~1]{PhysRevA.66.062311} the authors show that given any Hermitian matrix $\Delta_{1,2}$ with $\|\Delta_{1,2}\|_2 \leq 1$, the matrix $\mathbb{I}_{1,2} + \Delta_{1,2}$ is separable, where $\|\cdot\|_2$ is the 2-norm. Now, given any projector $\Pi_{1,2}$, $\|-\Pi_{1,2}\|_2 = |-1|\|\Pi_{1,2}\|_2 \leq 1$, and therefore setting $\Delta_{1,2} = -\Pi_{1,2}$ concludes the claim.
\end{proof}

\begin{lemma}
    Let $\tilde{\rho}_{1,2}$ be a POPT matrix with $\tr[\tilde{\rho}_{1,2}]=1$ and let $\sigma \succeq 0$. Then $\tr[\tilde{\rho}_{1,2}\sigma] \leq \tr[\sigma]$.
\label{lemma::tracebound}
\end{lemma}

\begin{proof}
    Since $\sigma \succeq 0 $ there exists a set $\{\ket{\phi_k}\rangle\}_k$, such that
\begin{equation}
    \sigma = \sum_{k} \alpha_k \ket{\phi_k}\rangle, \quad \alpha \geq 0 \ \forall \ k.    
\end{equation}
Since for every $k$, $\mathbb{I}-\ket{\phi_k}\rangle$ is separable following Lemma~\ref{Lemma::separabilitu_comp_proj}, $\tr[\tilde{\rho}_{1,2} \ket{\phi_k}\rangle ] \geq 0$. Therefore,
\begin{equation}
\begin{split}
    \tr[\tilde{\rho}_{1,2}(\mathbb{I}-\ket{\phi_k}\rangle)] &= \tr[\tilde{\rho}_{1,2}] - \tr[\tilde{\rho}_{1,2} \ket{\phi_k}\rangle]\\
    &=1 - \tr[\tilde{\rho}_{1,2} \ket{\phi_k} \rangle] \\
    &\implies \tr[\tilde{\rho}_{1,2} \ket{\phi_k}\rangle] \leq 1. \\
\end{split}   
\label{equpperb}
\end{equation}
Therefore, 
\begin{equation}
    \begin{split}
       \tr[\tilde{\rho}_{1,2}\sigma] = \tr\left[\tilde{\rho}_{1,2} \sum_{k} \alpha_k \ket{\phi_k} \rangle \right] = \sum_{k} \alpha_{k} \tr[\tilde{\rho}_{1,2} \ket{\phi_k}\rangle] \quad \mathrel{\stackrel{\makebox[0pt]{\mbox{\normalfont \small Eq~\ref{equpperb}}}}{\leq}} \quad \sum_{k} \alpha_{k} = \tr[\sigma].
    \end{split}
\end{equation}
\end{proof}

\begin{lemma}
    Every rank-1 unit-trace positive semi-definite matrix $\ket{\psi}\rangle_{1,2}$ is an element of $ \mathrm{Ext}[\mathcal{S}_{\mathcal{Q}_1} \boxtimes \mathcal{S}_{\mathcal{Q}_2}]$, where $\mathcal{S_Q}_{1,2}  \subseteq   \mathcal{S}_{\mathcal{Q}_1} \boxtimes \mathcal{S}_{\mathcal{Q}_2} \subseteq \mathcal{S}_{\mathcal{Q}_1} \otimes_{\text{max}} \mathcal{S}_{\mathcal{Q}_2}$.
    \label{lemma::rank1extremal}
\end{lemma}
\begin{proof}
    In~\cite[Theorem~14]{hansen2015extremal}, the authors show that $\ket{\psi}\rangle_{1,2}  \in \mathrm{Ext}\left[\mathcal{S}_{\mathcal{Q}_1} \otimes_{\text{max}} \mathcal{S}_{\mathcal{Q}_2}\right]$. Now assume that for some composition $\boxtimes$, of quantum subsystems such that $\mathcal{S_Q}_{1,2}  \subseteq   \mathcal{S}_{\mathcal{Q}_1} \boxtimes \mathcal{S}_{\mathcal{Q}_2} \subseteq \mathcal{S}_{\mathcal{Q}_1} \otimes_{\text{max}} \mathcal{S}_{\mathcal{Q}_2}$, there exists $\tau,\tau' \in \mathcal{S}_{\mathcal{Q}_1} \boxtimes \mathcal{S}_{\mathcal{Q}_2}$ such that for some $\lambda \in (0,1)$
    \begin{equation}
        \ket{\psi}\rangle_{1,2} = \lambda \tau + (1-\lambda) \tau'.
    \end{equation}
    Since, $\mathcal{S}_{\mathcal{Q}_1} \boxtimes \mathcal{S}_{\mathcal{Q}_2} \subseteq \mathcal{S}_{\mathcal{Q}_1} \otimes_{\text{max}} \mathcal{S}_{\mathcal{Q}_2}$, $\tau,\tau' \in \mathcal{S}_{\mathcal{Q}_1} \otimes_{\text{max}} \mathcal{S}_{\mathcal{Q}_2}$. And since $\ket{\psi}\rangle_{1,2}  \in \mathrm{Ext}\left[\mathcal{S}_{\mathcal{Q}_1} \otimes_{\text{max}} \mathcal{S}_{\mathcal{Q}_2}\right]$, this gives $\tau=\tau'=\ket{\psi}\rangle_{1,2}$.
\end{proof}

\begin{lemma}
    If $\Pi_{1,2} = \ket{\psi}\rangle_{1,2}$ is an element of $\mathcal{E}_{\mathcal{S}_{\mathcal{Q}_1} \boxtimes \mathcal{S}_{\mathcal{Q}_2}}$, then $\Pi_{1,2} = \ket{\psi}\rangle_{1,2} \in \mathrm{Ext}\left[\mathcal{E}_{\mathcal{S}_{\mathcal{Q}_1} \boxtimes \mathcal{S}_{\mathcal{Q}_2}}\right]$. 
    \label{projector_effect_extremal}
\end{lemma}
\begin{proof}
    Follows from the inclusion $\mathcal{E}_{\mathcal{S}_{\mathcal{Q}_{1,2}}} \supseteq \mathcal{E}_{\mathcal{S}_{\mathcal{Q}_1} \boxtimes \mathcal{S}_{\mathcal{Q}_2}}$.
\end{proof}

\begin{lemma}
    If the operational dimension of $(\mathcal{S}_{\mathcal{Q}_1},\mathcal{E}_{\mathcal{S}_{\mathcal{Q}_1}})$ and $(\mathcal{S}_{\mathcal{Q}_2},\mathcal{E}_{\mathcal{S}_{\mathcal{Q}_1}})$ are $d_1$ and $d_2$ respectively, then the operational dimension of $(\mathcal{S}_{\mathcal{Q}_1} \boxtimes \mathcal{S}_{\mathcal{Q}_2},\mathcal{E}_{\mathcal{S}_{\mathcal{Q}_1} \boxtimes \mathcal{S}_{\mathcal{Q}_2}})$, such that $\mathcal{S}_{\mathcal{Q}_1} \boxtimes \mathcal{S}_{\mathcal{Q}_2} \subseteq \mathcal{S}_{\mathcal{Q}_1} \otimes_{\text{max}} \mathcal{S}_{\mathcal{Q}_2}$, is $d_1d_2$.
    \label{lemma::multiplicityquantum}
\end{lemma}
\begin{proof}
    Let $d$ be the operational dimension of $(\mathcal{S}_{\mathcal{Q}_1} \boxtimes \mathcal{S}_{\mathcal{Q}_2},\mathcal{E}_{\mathcal{S}_{\mathcal{Q}_1} \boxtimes \mathcal{S}_{\mathcal{Q}_2}})$. Therefore there exists a set of extremal states $\{s_i\}_{i=1}^d$ and a measurement $\{e_i\}_{i=1}^d$, such that $\tr[e_is_j] = \delta_{i,j}$. Since for every $i$, $\tr[e_is_i]=1$ and $\tr[s_i]=1$, $\tr[e_i] \geq1$. Therefore,
    \begin{equation}
    \begin{split}
        d &\leq \sum_{i=1}^d\tr[e_i] = \tr\left[ \sum_{i=1}^d e_i\right] =\tr[\mathbb{I}_{1,2}] = d_1d_2.  \\       
    \end{split}        
    \end{equation}
One can show that this upper bound is tight by the following example. Consider the MPD set and the corresponding MDE measurement both characterised by, $\{\ket{\phi_k}\rangle_{k=1}^{d_1}\}$,  from $(\mathcal{S}_{\mathcal{Q}_1},\mathcal{E}_{\mathcal{S}_{\mathcal{Q}_1}})$ and similarly $\{\ket{\phi'_l}\rangle_{l=1}^{d_2}\}$ from $(\mathcal{S}_{\mathcal{Q}_2},\mathcal{E}_{\mathcal{S}_{\mathcal{Q}_2}})$. From this construct the set $\{\ket{\phi_k}\rangle \otimes\ket{\phi'_l} \}_{k,l=1}^{d_1,d_2}$. 
  This is a valid measurement because all elements are products. In addition, this is an extremal measurements since all effects are extremal, following Lemma~\ref{projector_effect_extremal}. This also represents a valid collection of extremal states, following Lemma~\ref{lemma::rank1extremal}. And finally the cardinality of this set is $d_1d_2$, which is the upper bound.
\end{proof}

\begin{lemma}
    Let $\{e_i\}_{i=1}^{d_1d_2}$ be an MDE measurement of $(\mathcal{S}_{\mathcal{Q}_1} \boxtimes \mathcal{S}_{\mathcal{Q}_2},\mathcal{E}_{\mathcal{S}_{\mathcal{Q}_1} \boxtimes \mathcal{S}_{\mathcal{Q}_2}})$ where $\mathcal{S_Q}_{1,2}  \subseteq   \mathcal{S}_{\mathcal{Q}_1} \boxtimes \mathcal{S}_{\mathcal{Q}_2} \subseteq \mathcal{S}_{\mathcal{Q}_1} \otimes_{\text{max}} \mathcal{S}_{\mathcal{Q}_2}$. Then $\tr[e_i]=1$ for all $i \in \{1,\dots,d_1d_2\}$.
    \label{lemma::unit_traceof_effects}
\end{lemma}

\begin{proof}
  Let $\{s_i\}_{i=1}^{d_1d_2}$ be the corresponding MPD set. From Lemma~\ref{lemma::tracebound}, we have that for every $i \in \{1,\dots,d_1d_2\}$, $\tr[e_i\tilde{\rho}] \leq \tr[e_i]$ for every $\tilde{\rho} \in \mathcal{S}_{\mathcal{Q}_1} \boxtimes \mathcal{S}_{\mathcal{Q}_2}$. Additionally, we also have that $\tr[e_is_i]=1$. This implies $\tr[e_i] \geq 1$. Now since there are $d_1d_2$ effects in the measurement and $\tr[\sum_{i=1}^{d_1d_2}e_i] = d_1d_2$, we get $\tr[e_i]=1$ fro all $i \in \{1,\dots,d_1d_2\}$. 
\end{proof}

\QuantumTensorProduct*

\begin{proof}
    We know that quantum theory admits mutual superposition. Here we will show that no composition, $\mathcal{S}_1 \boxtimes \mathcal{S}_2$, such that $\mathcal{S_Q}_{1,2}  \subset   \mathcal{S}_{\mathcal{Q}_1} \boxtimes \mathcal{S}_{\mathcal{Q}_2} \subseteq \mathcal{S}_{\mathcal{Q}_1} \otimes_{\text{max}} \mathcal{S}_{\mathcal{Q}_2}$ will admit mutual superposition. Let $\tilde{\rho}_{12} \in  \mathcal{S}_{\mathcal{Q}_1} \boxtimes \mathcal{S}_{\mathcal{Q}_2} \setminus \mathcal{S_Q}_{1,2}$. Since $\tilde{\rho}$ is not positive semidefinite, it has at least one eigenvector $\ket{\psi_{\tilde{\rho}}}_{1,2}$, such that
    \begin{equation}
        \tr\left[ \tilde{{\rho}}  \ket{\psi_{\tilde{\rho}}}\rangle_{1,2} \right] < 0,
        \label{equation::witness}
    \end{equation}
and from the fact that $\tilde{\rho}$ is POPT, $\ket{\psi}\rangle$ must also be non-separable across the partition $1|2$. Since normalisation does not change the inequality above we abuse notation to denote $\ket{\psi_{\tilde{\rho}}}\rangle_{1,2}$ as its normalised state from here.  From Lemma~\ref{lemma::rank1extremal}, we have that $ \ket{\psi_{\tilde{\rho}}}\rangle_{1,2} \in \mathrm{Ext}[\mathcal{S}_{\mathcal{Q}_1} \boxtimes \mathcal{S}_{\mathcal{Q}_2}]$. Next, we show that the state $ \ket{\psi_{\tilde{\rho}}}\rangle_{1,2}$ is a superposition. To do this, consider the Schmidt decomposition 
\begin{equation}
    \ket{\psi_{\tilde{\rho}}}_{1,2} = \sum_{i,j=1}^{m,n} \sqrt{p_{i,j}} \ket{\alpha_i}_1 \otimes \ket{\beta_j}_2,
\end{equation}
where $p_{i,j} \in (0,1)$ for any $i,j \in \{(1,1), \dots,(m,n)\}$ and $\sum_{i,j} p_{i,j} =1$. Let $\{\ket{\alpha_i}_1\}_{i=1}^{d_1}$ and $\{\ket{\beta_j}_1\}_{j=1}^{d_2}$ be the completion of the orthogonal sets of vectors $\{\ket{\alpha_i}_1\}_{i=1}^{m}$ and $\{\ket{\beta_j}_1\}_{j=1}^{n}$ to orthonrmal bases. Following the example in the proof of Lemma~\ref{lemma::multiplicityquantum}, the collection  $\{\ket{\phi_k}\rangle \otimes\ket{\phi'_l} \}_{k,l=1}^{d_1,d_2}$ then forms an MPD set and also its corresponding MDE measurement. Since $ \ket{\psi_{\tilde{\rho}}}\rangle_{1,2}$ is probablistic with respect to this measurement, it is a superposition of some subset of $\{\ket{\phi_k}\rangle \otimes\ket{\phi'_l} \}_{k,l=1}^{d_1,d_2}$.

Next, we will show that $ \ket{\psi_{\tilde{\rho}}}\rangle_{1,2}$ cannot be a part of an MPD set. To be a part of an MPD set, one needs at least one extremal effect $f$ of $\mathcal{S}_{\mathcal{Q}_1} \boxtimes \mathcal{S}_{\mathcal{Q}_2}$ such that $\tr[f \ket{\psi_{\tilde{\rho}}}\rangle_{1,2}] = 1$. One can show that $f$ is indeed unique and must be  $\ket{\psi_{\tilde{\rho}}}\rangle_{1,2}$ itself. Since $f$ must be a valid effect of $\mathcal{S_Q}_{1,2}$ it is positive semidefinite with a spectral decomposition
\begin{equation}
    f = \sum_{k} \lambda_k \ket{\phi_k}\rangle, \quad \lambda_k \in [0,1] \ \forall\ k, \quad  \sum_{k} \lambda_k = 1, 
\end{equation}
where the second equality holds because $f$ is a part of an MDE measurement and thereofore, following Lemma~\ref{lemma::unit_traceof_effects}, $\tr[f]=1$.  With this one gets
\begin{equation}
    \begin{split}
        \tr[f\ket{\psi_{\tilde{\rho}}}\rangle_{1,2}] &= \sum_{k} \lambda_k \tr[\ket{\phi_k}\rangle \ket{\psi_{\tilde{\rho}}}_{1,2}] \leq \sum_{k} \lambda_k = 1.
    \end{split}
\end{equation}
For equality to hold, we require $\tr[\ket{\phi_k}\rangle \ket{\psi_{\tilde{\rho}}}_{1,2}] = 1$ for all $k$. Since, $f$ is Hermitian, $\{\ket{\phi_k}\rangle\}_k$ is an orthogonal set. This implies that $\lambda_k = 1$ for exactly one $k$ for which $\ket{\phi_k}\rangle = \ket{\psi_{\tilde{\rho}}}\rangle_{1,2}$. Finally since $\tilde{\rho}$ is an allowed state, from Equation~\ref{equation::witness}, $f$ cannot be an allowed effect.

\end{proof}

\subsection{Uncertainty implies Superposition: Proof of Lemma~\ref{Thm::UncertaintyImpliesSuperposition}} \label{Appendix: Proof Uncertainty implies Superposition}

\UncImpliesSup*

\begin{proof}
    (i). If there exists MDEs $\mathsf{M},\mathsf{N}$ such that $\Det(\mathsf{M}) \cap \Det(\mathsf{N}) = \emptyset$, then the set of extremal states can be partitioned into three disjoint sets:
    \begin{equation}
        \mathrm{Ext}[\mathcal{S}] = \Det(\mathsf{M}) \cup  \Det(\mathsf{N}) \cup \mathrm{Ext}[\mathcal{S}] \setminus \left( \Det(\mathsf{M}) \cup  \Det(\mathsf{N})  \right).
    \end{equation}
    If $s \in \Det(\mathsf{M})$, then $s$ is in superposition of some $\mathrm{D}' \subseteq \Det(\mathsf{N})$. The argument is symmetric for if $s \in  \Det(\mathsf{N})$. Now if $s$ is an element of the remaining set, then it is probabilistic w.r.t.\ both $\mathsf{M}$ and $\mathsf{N}$ and therefore also a superposition. This covers all extremal states. It follows that $(\mathcal{S},\mathcal{E})$ exhibits complete superposition, hence superposition exists in the theory.

    (iii). According to mutual superposition, for every MPD set $\mathsf{D}$, every state in $\mathrm{Ext}[\mathcal{S}] \setminus \mathsf{D}$ is a superposition of (some of) the elements of $\mathsf{D}$. If preparational uncertainty exists, then the set $\mathrm{Ext}[\mathcal{S}] \setminus \Det(\mathsf{M})$ is non-empty for any MDE $\mathsf{M}$---actually, for this to hold true, preparation uncertainty is not necessary: the weaker condition $|\mathrm{Ext}[\mathcal{S}]| >n $ is sufficient, where $|\cdot|$ denotes the cardinality of the set, and $n$ is the maximum number of perfectly distinguishable states (i.e.\ the size of an MPD set). Due to outcome sharpness of MDEs, we can write $\mathrm{Ext}[\mathcal{S}] \setminus \Det(\mathsf{M}) =  \mathrm{Ext}[\mathcal{S}] \setminus \mathsf{D}_{\mathsf{M}} \neq \emptyset$, where $\mathsf{D}_{\mathsf{M}}$ is the unique MPD set distinguished by $\mathsf{M}$. All elements of this set are probabilistic w.r.t.\ $\mathsf{M}$, hence a superposition of (elements of) $\mathsf{D}_{\mathsf{M}}$. This holds for all MPD sets, hence uniform superposition holds.

\end{proof}

\subsection{Restating Uncertainty in terms of Superposition: Proof of Corollary~\ref{Lem::UncertaintyAsSuperposition}} \label{Appendix: Restating uncertainty wrt superposition}

\UncertaintyAsSuperposition*

\begin{proof}
    (i). If there exists MDEs $\mathsf{M},\mathsf{N}$ such that $\Det(\mathsf{M}) \cap \Det(\mathsf{N}) = \emptyset$, then there must exist an MPD set $\mathsf{D} \subseteq \Det(\mathsf{M})$ composed entirely of states in a superposition of an MPD set $\mathsf{D}' \subseteq \Det(\mathsf{N})$. To show the converse implication, consider outcome sharpness of MDEs: we have that $\mathsf{D}= \Det(\mathsf{M})$ and $\mathsf{D}' = \Det(\mathsf{N})$. Therefore, requiring, say, $\mathsf{D}$ to be composed \textit{entirely} of states in a superposition of $\mathsf{D}'$ implies that $\Det(\mathsf{M}) \cap \Det(\mathsf{N}) = \emptyset$, hence preparational uncertainty exists in the GPT. 

    (ii). Weak uncertainty implies that for all $ \mathsf{M} $ there exists  $\mathsf{N} $ such that $\Det(\mathsf{M}) \cap \Det(\mathsf{N}) = \emptyset$. Therefore, for all MPDs $\mathsf{D} \subseteq \Det(\mathsf{M})$ there exists some MPD set $ \mathsf{D}' \subseteq \Det(\mathsf{N})$ such that each element of $\mathsf{D}'$ is probabilistic with respect to $\mathsf{M}$, hence a superposition of the elements of $\mathsf{D}$. To show the converse, we again assume outcome sharpness of MDEs: the requirement that for every MPD set $\mathsf{D}$ there is some $\mathsf{D}'$ composed entirely of states in a superposition of $\mathsf{D}$ immediately implies that for all MDEs $\mathsf{M}$ there exists an MDE $\mathsf{N}$ such that $ \Det(\mathsf{M}) \cap \Det(\mathsf{N}) = \emptyset$.

\end{proof}

\section{Non-multiplicative Operational Dimension: Composite Regular Pentagons} \label{Appendix::MinProductPentagons}
Let $(\mathcal{S}_5, \mathcal{E}_5)$ denote the state and effect spaces of a regular pentagon, cf.\ Appendix \ref{Appendix:RegularNgons}. This system possesses an operational dimension of two (see Figs.\ \ref{Fig::PentagonHexagon}-\ref{Fig::PentagonSuperposition}); in particular, the inner product $e_5(j)(\omega_5 (i))$ vanishes when $i-j= \pm 2$ (mod $5$).

Consider an arbitrary no-signalling composition $\mathcal{S}_{1,2} = \mathcal{S}_5 \boxtimes \mathcal{S}_5$ of two pentagon state-spaces, and let $\mathcal{E}_{1,2}$ be the associated maximal effect space. Products of extremal states and effects are extremal elements of $\mathcal{S}_{1,2}$ and $\mathcal{E}_{1,2}$, respectively. If the operational dimension were multiplicative, the cardinality of MPD sets in $(\mathcal{S}_{1,2}, \mathcal{E}_{1,2})$ would be four. However, we can construct a PD set of five product states that can be discriminated by an extremal measurement consisting solely of product effects. Defining $S_i \equiv \omega_5 (i) \otimes \omega_5 (2i-1)$ and $E_j \equiv e_5 (j) \otimes e_5 (2j-1)$, the set $\mathsf{D} = \{ S_1, S_2, S_3, S_4, S_5\}$ is perfectly distinguished by the measurement $\mathsf{M} = \{ E_1, E_2, E_3, E_4, E_5\} $. Note that
\[
    \sum_{j=1}^5 E_j = \frac{1}{5}  \sum_{j=1}^5
\begin{pmatrix}
r \cos\theta_j \\
r \sin\theta_j \\
1
\end{pmatrix}
\otimes
\begin{pmatrix}
r \cos\phi_j \\
r \sin\phi_j \\
1
\end{pmatrix},
\]
where $\theta_j = \frac{2\pi j}{5}$ and $\phi_j = \frac{2\pi(2j-1)}{5}$. 
By symmetry of the regular pentagon, one has
\[
\sum_{j=1}^5 \cos\theta_j = \sum_{j=1}^5 \sin\theta_j 
= \sum_{j=1}^5 \cos\phi_j = \sum_{j=1}^5 \sin\phi_j = 0,
\]
and similarly all mixed trigonometric terms vanish upon summation. Hence only the constant component survives, yielding
\[
\sum_{j=1}^5 E_j =
\begin{pmatrix}
0\\
0\\
1
\end{pmatrix}
\otimes
\begin{pmatrix}
0\\
0\\
1
\end{pmatrix}
= u \otimes u \, .
\]

By construction, $E_j(S_j)=1$ for all $j$. For the off-diagonal terms $E_j(S_i)$, where $i \neq j$, we consider two cases: (\textit{i}) if $i-j = \pm 2$, then the first factor in $E_j(S_i)$ vanishes, i.e., $e_5(j)(\omega_5 (i)) = 0$; (\textit{ii}) if $i-j =\pm 1$, then $(2i-1)-(2j-1) = \pm 2$ causing the second factor in $E_j(S_i)$ to vanish, i.e., $e_5(2j-1)(\omega_5 (2i-1)) = 0$.

Product states and effects are therefore sufficient to show that the operational dimension of an arbitrary tensor product of two regular pentagons strictly exceeds the product of the dimensions of its individual subsystems.

\section{Connections to Related Work} \label{Appendix: Related work}

In this section, we compare our operational definition of superposition (Definitions~\ref{Def::Superposition}-\ref{Def::SuperpositionPrinciples}) with previous approaches presented in \cite{dariano_classical_2020, aubrun_entanglement_2022}.

\subsection{Superposition in D'Ariano \textit{et al.}, 2020}

In \cite{dariano_classical_2020}, D'Ariano et al.\ introduce a hierarchy of three superposition principles. Similar to our construction, they consider an MPD set $\mathsf{D}=\{s_i\}_{i=1}^d$ and some measurement $\mathsf{M}=\{e_i\}_{i=1}^d$ that distinguishes them (though in their case $\mathsf{M}$ is not necessarily an MDE). Let $\mathsf{p}=\{ p_i \}_{i=1}^d$ denote a probability distribution and $r$ an extremal state, a GPT admits: (i) \textit{ultraweak superposition} iff for every choice of $\mathsf{p}$, there exist $\mathsf{M}$ and $r$ such that 
    \begin{equation} \label{eq:DAriano superposition equation}
        p_i= e_i(r) \,\,\,\, \text{for all } i=1,\dots , d
    \end{equation}
holds; (ii) \textit{weak superposition} iff for every choice of $\mathsf{p}$ and $\mathsf{M}$, there exists $r$ such that Eq.\ \eqref{eq:DAriano superposition equation} holds; (iii) \textit{strong superposition} iff for every choice of $\mathsf{p}$, there exists $r$ such that, for every choice of $\mathsf{M}$, Eq.\ \eqref{eq:DAriano superposition equation} holds.

Superposition is not introduced as a relational property between a state and a PD set; instead, it is characterised via global structural properties of the state and effect spaces. Because the proposed definitions rely on arbitrary choices for the probability vector $\mathsf{p}$, they only apply to GPTs with infinitely many extremal states. Consequently, in their approach superposition remains undefined for theories such as Spekkens' model and regular $n$-gons.

Using their notation, $(\mathcal{S},\mathcal{E})$ exhibits superposition according to our Definition~ \ref{Def::Superposition} if there exists a \textit{non-dispersion-free} distribution $\mathsf{p}$, an \textit{MDE} $\mathsf{M}$, and an extremal state $r$ such that Eq.\ \eqref{eq:DAriano superposition equation} holds. Therefore, our basic notion of superposition is implied by each of their definitions; nonetheless, our principles of complete, uniform, and mutual superposition (Definition~\ref{Def::SuperpositionPrinciples}) highlight separate structural features of the state space. Crucially, both approaches share the core premise that the operational hallmark of superposition---the operational phenomenon that defines it---is the observed probabilistic behaviour of extremal states relative to measurements that are maximally distinguishable.

The authors of \cite{dariano_classical_2020} then show that classical theories (defined by a simplicial state-space whose vertices are perfectly distinguishable) do not exhibit any of their superposition principles. Simplicial theories that are non-classical (i.e.\ in which the perfect distinguishability of the vertices is violated, thereby preclusing state tomography) fail to satisfy weak (and strong) superposition. However, an explicit example of a non-classical simplicial theory satisfying ultraweak superposition remains unknown.

\subsection{Superposition in Aubrun \textit{et al.}, 2022}
\label{Appendix::PlavalaSup}
The authors of \cite{aubrun_entanglement_2022} discuss how an operational notion of superposition in GPTs can ``\textit{stem from the comparison between the two ensembles $\{ \ket0, \ket1 \}$ and $\{ \ket+, \ket- \}$, where $\ket \pm = (\ket0 \pm \ket1)/\sqrt{2}$. The ensembles, in fact, when mixed with equal weights, give rise to the same density matrix; further, identifying the state unambiguously via a measurement is possible or impossible depending on whether the ensemble is revealed before or after the measurement is carried out.}''.

Based on this idea, they introduce the property of \textit{strong incompatibility} as ``\textit{a way to identify the existence of superpositions in an arbitrary GPT}''. A GPT satisfies strong incompatibility if there exist (unnormalised) states $0 \neq s_0, s_1, s_+, s_-$ and (unnormalised) effects $e_0, e_1, e_+, e_-$ such that (i) $s_0+s_1 = s_+ + s_-$; (ii) $e_o(s_1)=e_1(s_0)=e_+(s_-)=e_-(s_+)=0$; and (iii) $e_i+e_j$ is strictly positive, for all $i \in \{0,1\}$, $j\in \{+,-\}$. The connection to superposition arises from the parallel between the ensembles $\{s_0,s_1\}$ and $\{ s_+, s_-\}$ and the ensembles $\{ \ket0, \ket1 \}$ and $\{ \ket+, \ket- \}$, where the effects $e_0,e_1,e_+,e_-$ correspond to the rank-1 projections generated by the respective vectors.

Strong incompatibility (which is assumed to be equivalent to the existence of superposition) is shown to boil down to the non-classicality of the theory, and thus with the possibility of combining two of more of such systems with a tensor product that allows entanglement. While this notion of superposition also applies to GPTs with finitely many extremal state, it is weaker than those provided in \cite{dariano_classical_2020} and in Definition~\ref{Def::Superposition} of this paper. For instance, the gbit system exhibits this form of superposition; to show this, pick diagonal vertices for the pairs of $s_i$ states and equal mixtures of extremal effects for the $e_i$ effects. 
Furthermore, as in \cite{dariano_classical_2020}, superposition here is \textit{not} treated as a relational property between states, as is conventional in quantum theory.

The version of superposition introduced in \cite{aubrun_entanglement_2022} is implied by our own. What is more, their notion proves to be equivalent to a weakened version of Definition~\ref{Def::Superposition}, in which the state $r$ is allowed to be mixed, provided it is not a mixture of the elements of $\mathsf{D}$.
This weaker notion of superposition is formalised as follows:

\begin{definition}[Weak Superposition] 
    Let $(\mathcal{S},\mathcal{E})$ be a pair of state and effect spaces. Let $\mathsf{D} \coloneqq\{s_i\}_{i=1}^n$ be an MPD set and $\mathsf{M} \coloneq \{e_i\}_{i=1}^n$ an MDE measurement of $\mathsf{D}$. Let $r \notin \mathrm{ConvHull}[\mathsf{D}]$ be a state with the property that for some $\mathsf{M}' \subseteq \mathsf{M}$, $ e_j(r)\in (0,1)$ holds for all $e_j \in \mathsf{M}'$, and $\sum_{e_j \in \mathsf{M}'} e_j(r) =1$. Then, $r$ is a superposition of the states in the set $\mathsf {D}' \coloneqq \{s \in \mathsf D | e_j(s) =1, e_j \in \mathsf{M}' \}$. 
    \label{Def::WeakSuperposition}
\end{definition}

\begin{theorem} \label{Theorem::weak superposition equivalent to nonclassicality}
Let $(\mathcal{S},\mathcal{E})$ be a state and effect pair of a GPT. Weak superposition (as per Definition~\ref{Def::WeakSuperposition}) exists if and only if $(\mathcal{S},\mathcal{E})$ is non-classical.
\end{theorem}
\begin{proof}
    If $(\mathcal{S},\mathcal{E})$ exhibits weak superposition as per Definition~\ref{Def::WeakSuperposition}, then $\mathcal{S} \setminus \text{ConvHull}(\mathsf{D}) \neq \emptyset$, hence it is non-classical, since in a classical GPTs, $\mathcal{S}= \text{ConvHull}(\mathsf{D})$, where $\mathsf{D}$ is the (unique) MPD set formed by the vertices.

    For the converse, suppose $\mathcal{S}$ is not a simplex. Then, for any choice of MPD $\mathsf{D}=\{s_1,\dots,s_n\}$, there is at least one $t \in \text{Ext}[\mathsf{S}] \setminus \mathsf{D}$. Fix an MDE $\mathsf{M}=\{e_i \}_{i=1}^n$ that distinguishes $\mathsf{D}$: if $t$ is probabilistic with respect to $\mathsf{M}$, then $(\mathcal{S},\mathcal{E})$ exhibits superposition as per Definition~\ref{Def::Superposition}, hence also the weaker form of Definition~\ref{Def::WeakSuperposition}. If $t \in \Det(\mathsf{M})$, then $e_j(t)=1$ for exactly one $j\in \{1,\dots,n\}$, with $e_i(t)=0$ for all $i\neq j$. Consider the mixture $\eta_p=pt+(1-p)s_i$, with $i\neq j$. We have that $\eta_0\in \text{ConvHull}(\mathsf{D})$, $\eta_1 \notin \text{ConvHull}(\mathsf{D})$ and that $e_j(\eta_p)\in (0,1)$ as long as $p\neq 0,1$. By the hyperplane separation theorem, there exists $p' \in (0,1)$ such that $\eta_p \in \text{ConvHull}(\mathsf{D})$ for all $p\leq p'$ and $\eta_p \notin \text{ConvHull}(\mathsf{D})$ for all $p > p'$. Hence there exist non-trivial mixtures of $t$ and $s_i$ which are probabilistic with respect to $\mathsf{M}$ and lie outside $\text{ConvHull}(\mathsf{D})$. Hence $(\mathcal{S},\mathcal{E})$ exhibits weak superposition.
\end{proof}

\end{document}